\documentclass[journal]{IEEEtran}
\usepackage{color}
\usepackage{amssymb} %for symbols
\usepackage{url}
\usepackage{mathtools}
\usepackage{amsthm} %for theorem and proof

\usepackage{tabularx,ragged2e,booktabs,caption}

\usepackage{graphicx}
\DeclareGraphicsExtensions{.pdf,.png,.jpg}
\usepackage{subcaption}

\begin{document}
	
	%\title{Throughput-Delay Tradeoffs in Content-Centric Wireless Networks}
	\title{Throughput and Delay Scaling of Content-Centric\\ Ad Hoc and Heterogeneous Wireless Networks}
	
	\author{Milad~Mahdian,~\IEEEmembership{Member,~IEEE,}
		Edmund~Yeh\thanks{This work was supported by National
			Science Foundation grant CNS-1423250
			and a Cisco Systems research grant.},~\IEEEmembership{Senior Member,~IEEE}
	}
	
	%
	%% The paper headers
	%\markboth{Journal of \LaTeX\ Class Files,~Vol.~11, No.~4, December~2012}%
	%{Shell \MakeLowercase{\textit{et al.}}: Bare Demo of IEEEtran.cls for Journals}
	%
	
	% make the title area
	\maketitle

	\newtheorem{theorem}{Theorem}
	\newtheorem{lemma}{Lemma}
	\newtheorem{defi}{Definition}

	\begin{abstract}
We study the throughput and delay characteristics of wireless caching networks, where users are mainly interested in retrieving content stored in the network, rather than in maintaining source-destination communication. Nodes are assumed to be uniformly distributed in the network area. Each node has a limited-capacity content store, which it uses to cache contents.  We propose an achievable caching and transmission scheme whereby requesters retrieve content from the caching point which is closest in Euclidean distance.  We establish the throughput and delay scaling of the achievable scheme, and show that the throughput and delay performance are order-optimal within a class of schemes.  We then solve the caching optimization problem, and evaluate the network performance for a Zipf content popularity distribution, letting the number of content types and the network size both go to infinity.  Finally, we extend our analysis to heterogeneous wireless networks where, in addition to wireless nodes, there are a number of base stations uniformly distributed at random in the network area. We show that in order to achieve a better performance in a heterogeneous network in the order sense, the number of base stations needs to be greater than the ratio of the number of nodes to the number of content types. Furthermore, we show that the heterogeneous network does not yield performance advantages in the order sense if the Zipf content popularity distribution exponent exceeds 3/2.
	\end{abstract}
	\section{Introduction and Related Work}
	\IEEEPARstart{T}{wo} fundamental trends in networking are: first, the bulk of network traffic today, and of
	its projected enormous growth, consists mainly of content disseminated to multiple users. Second, network
	content is accessed increasingly in wireless environments.  A basic problem, of both theoretical and practical interest,
	is the characterization of performance and scaling in large-scale wireless networks for content distribution.  This paper
	addresses this key question.  We focus on the well-known random wireless network model, where nodes are uniformly distributed in a 
	network area.  Rather than assuming a wireless communication network consisting of source-destination pairs, however, we investigate a wireless caching network infrastructure where users are mainly interested in retrieving content stored in the network.  Combining caching schemes  with the proposed request forwarding, we derive the throughput and delay scalings of the content-centric wireless network and solve the caching optimization problem.  We then extend our analysis to heterogeneous wireless networks with base stations as well as wireless nodes. 
	
	As the number of users of wireless technology continues to grow exponentially, the scaling behavior of wireless networks has been of wide interest.  Gupta and Kumar \cite{gupta} pioneered this study within the context of wireless communication networks consisting of source-destination pairs.  They focus on a random network model where $n$ nodes are distributed independently and uniformly on a unit disk. 
	Each node has a randomly chosen destination node and can transmit at $W$ bits per second provided that the interference is sufficiently small.  Each node can simultaneously serve as a source, a destination, and as a relay for other source-destination pairs.  It was shown \cite{gupta} that the per-source-destination-pair throughput scales as $\Theta(1/\sqrt{n\log n})$,\footnote{We use the
		following notation. We say $f(n) = O(g(n))$ if there exists $n_0 > 0$ and a constant $M$ such that
		$|f(n)| \leq M|g(n)|~\forall n \geq n_0$. We say $f(n) = o(g(n))$ if for any constant $\epsilon>0$
		there exists $n(\epsilon)> 0$ such that $|f(n)| \leq \epsilon|g(n)|~\forall n \geq n(\epsilon)$.
		We say $f(n) = \Omega(g(n))$ if $g(n) = O(f(n))$, and $f(n)=\omega(g(n))$ if $g(n)=o(f(n))$.
		Finally, we say $f(n) = \Theta(g(n))$ if $f(n) = O(g(n))$ and $f(n) = \Omega(g(n))$.} where $n$ is the number of wireless nodes in the network.
	Subsequent work was devoted to characterizing the tradeoff between throughput and delay~\cite{lin,neely,mobility, elgamal, elgamalC, anycast, towsley,kulkarni}.  
	%This issue causes wireless networks not sustain long multi-hop communication schemes. 
	In particular, El Gamal et al. \cite{elgamal, elgamalC} study both static and mobile wireless networks, and  {show that the optimal per-node throughput and network delay for the static wireless network scenario are  
	%\begin{equation}
	%D(n) = \Theta(nT (n)), \label{elgamal-tradeoff}
	%\end{equation}
	$\lambda(n) = \Theta( 1/(n\sqrt{a(n)}) )$ and $D(n) = \Theta( 1/\sqrt{a(n)} )$, respectively,  
	 where $n$ is the number of wireless nodes in the network, and $a(n)$ is the appropriately chosen cell size such that $a(n) = \Omega (\log n/n)$.}
	
	In \cite{towsley}, Liu et al., extend the ad hoc network model to a hybrid model in which a sparse number of base stations are placed in the wireless network. They show that for a hybrid network
	of $n$ nodes and $m$ base stations, if $m = o(\sqrt{n})$, the benefit of including additional base stations
	on capacity is insignificant in the order sense. However, for $m = \Omega(\sqrt{n})$, 
	the throughput capacity increases linearly with the number of
	base stations, improving the scaling of the network's performance over the pure
	ad hoc case.

	 As shown in these papers, the throughput of wireless networks scales poorly with number of users. In general, for a static wireless network, the maximum common rate sustainable for all flows in the network scales inversely with the number of hops. In \cite{mobility}, the authors show that mobility can improve the throughput of wireless networks. In particular, they show that direct communication between sources and destinations
	 alone cannot achieve high throughput. They propose a two-hop scheme in which the per-node throughput is $\Theta(1)$.  This result, however, comes with the price of large delays.  Specifically, the delay associated with their scheme is later shown to be $\Theta(n\log n)$. 
	  In \cite{zhening}, network coding is used to improve the delay of mobile wireless networks. By employing Reed-Solomon
	  codes, the authors improve the delay of the two-hop scheme in \cite{mobility} from $\Theta(n\log n)$ to $\Theta(n)$.

	In wireless networks running popular applications such as on-demand video and web browsing, caching content objects closer to requesters can significantly decrease the number of required hops, and has the potential to substantially improve throughput and delay scalings. 
%	\color{red}The host-based design of IP network architecture does not allow a specific piece of data to have an independent existence outside a conversation between sources and destinations. In order for IP to exploit caching, in particular web caching, content objects need to be mapped to host locations using complicated network configurations and services \cite{ndn}\color{black}. 
	Recently, new content-centric networking architectures such as Named Data Networking (NDN) \cite{ndn} and Content-Centric Networking (CCN) \cite{Jacobson} have been developed to more directly enable efficient content distribution using caching.  
	%These architectures focus on the efficient delivery of content to requesters, instead of maintaining conversations between sources and destinations. 
	
	Given the above, a natural and important problem is the characterization of performance and scaling in large-scale wireless caching networks.  The problem has received attention recently in \cite{asymptotic,italian}.  In \cite{asymptotic}, asymptotic properties of the joint delivery and replication problem in a static grid-based wireless network with multi-hop communication and caching are presented.  The objective here is the minimization of average link capacity subject to content replication constraints.  Scaling laws for link capacities are derived, with the content popularity following a Zipf distribution. 
	%In their network model, wireless nodes are located statically on a grid, while in our model, nodes are randomly distributed in the network area. In addition, we have managed to calculate practical performance parameters, delay and throughput as well as their tradeoff of the proposed network. It is worth noting that our network model is very similar to the traditional framework, which helps comparing the proposed scheme with the traditional S-D pairs wireless networks easier. Moreover, we study the heterogeneous cellular network model in this framework, which to our best of knowledge is the first attempt to analyze such a network paradigm.
	
	The paper~\cite{italian} derives the throughput and delay performance of content-centric mobile ad-hoc networks under various mobility models on a random geometric graph, for Zipf content popularity distributions.  The paper makes the assumption that at any given time, each node has at most one pending content request in the network.  It further considers a request model in which the relation between the throughput and delay is pre-determined as $\lambda=\frac{1}{\bar{I}+\bar{D}}$, where $\lambda$ is the 
	average request throughput, $\bar{D}$ is the average request delay, and $\bar{I}$ is the average time between consecutive content requests~\cite{italian}.  
	
	%However, in our paper, we assume infinite backlogged content requests at each node, and, we derive the trade off of the throughput and delay of the proposed scheme. 

	In~\cite{cedric}, the asymptotic throughput capacity of content-centric wireless networks is studied under the assumption that a constant number of content objects with similar popularity are requested and cached with limited lifetime by network users.  By computing the average lifetime of the cached content objects of each user, the network throughput is derived for both the grid and random network models.
	
	In \cite{femto}, a content placement problem in a wireless femto-cellular network using \emph{helper} nodes is studied.  The paper considers a one-hop communication scheme where nodes are connected to a set of helper nodes according to a bipartite graph.  Each node is also connected to the base station.  The paper focuses on the minimization of the average total downloading delay for a given content popularity distribution and network topology.  The authors show that the uncoded optimal file assignment is NP-hard, and demonstrate a greedy strategy with performance which is provably within a factor 2 of the optimum.  
	
	The authors of \cite{D2D-usc} analyze base-station-assisted device-to-device wireless networks with caching capability.  They examine a cellular grid network model in which communication among wireless nodes or between wireless nodes and the base station is limited to one hop, and derive the asymptotic throughput-outage tradeoff for the network model.

	Finally, the paper \cite{vip} develops a systematic framework to solve the fundamental problem of jointly optimizing interest request forwarding and dynamic cache placement and eviction, for arbitrary network topologies and content popularity distributions. 

In this paper, we characterize the throughput and delay scaling behavior of wireless caching networks, using the random geometric model as studied in \cite{gupta},  \cite{elgamal} and in many related papers (previously within the context of traditional source-destination communication networks).  We assume that contents follow a general popularity distribution, and that each node has a limited-capacity content store, which it uses to cache contents according to a proposed caching scheme.   Users employ multi-hop communication to retrieve the requested content from content stores caching the requested object.  

We propose an achievable caching and transmission scheme whereby holders of each content item are independently and uniformly distributed in the network area, and transmission proceeds according to a multi-hop, TDM, cellular scheme in which requesters retrieve content from the holder which is closest in Euclidean distance.  We establish the throughput and delay scaling of the achievable caching/transmission scheme, and show that the throughput and delay performance are order-optimal within a class of schemes. 

The per-node throughput $\lambda(n)$ and network delay $D(n)$ of the proposed achievable scheme is shown to satisfy\footnote{We say an event holds with high probability (w.h.p.) if the event occurs with probability 1 as $n$ goes to infinity.}
	\begin{equation}
	D(n)\lambda(n)= \Theta((na(n))^{-1} )\qquad\; w.h.p.  \label{eq:ccntradeoff}
	\end{equation}
It can be seen from \eqref{eq:ccntradeoff} that one can simultaneously increase the throughput while decreasing delay, for a given $n$ and $a(n)$.  This is accomplished by intelligently designing the caching and transmission scheme to decrease the number of transmissions and the accompanying interference.  

Next, we optimize the caching strategy to simultaneously minimize the average network delay and maximize the network throughput.  Using the optimal caching strategy, we evaluate the network performance under a Zipf content popularity distribution. 
	
Finally, we investigate heterogeneous wireless networks where, in addition to wireless nodes, there are a number of base stations uniformly distributed at random in the network area.  We show the proposed model and optimization approach can be naturally extended to the heterogeneous case.  The solution of the content placement optimization problem shows that the number of base stations needs to be greater than the ratio of the number of nodes to the number of content types in order to achieve a better performance in a heterogeneous network in the order sense. For the case where the number of content objects is greater than the number of wireless nodes, this condition reduces to having at least one base station in the network. In addition, we show that for the Zipf content popularity distribution with exponent $\alpha\geq 3/2$, the performance of the wireless ad hoc network is of the same order as for the heterogeneous wireless network, independent of number of base stations.
	
In contrast to related work, this paper offers the following unique contributions.   First, our paper uses the well-known random dense geometric network model, which was used in many previous papers on throughput and delay scaling in traditional source-destination wireless communication networks (e.g. \cite{gupta} and \cite{elgamal}).  This allows for a more direct performance comparison between wireless communication networks and content-centric wireless networks.  Specifically, this paper clearly shows that caching in wireless content-centric networks allows us to increase the throughput and decrease delay simultaneously.  Second, in contrast to related work, our paper demonstrates an achievable caching and transmission scheme and at the same time shows that the throughput and delay performance of the achievable scheme is optimal within a class of schemes.  Third, our paper is the first to characterize the throughput and delay scaling in heterogeneous wireless content-centric networks.
	\color{black}
	\section{Network Model}
	\label{model}
	
	We analyze a content-centric wireless network model where $n$ nodes are independently and uniformly distributed over a unit-sized torus.  From these nodes originate requests for content objects.  There are $M$ distinct content objects, where $M$ scales as $n^\beta$, $0<\beta<1$. Note that we assume $\beta<1$ in order for the network to have sufficient memory to store at least one copy of each content object.
	All content objects are assumed to have the same size.  Each node is assumed to have a local cache, named the \emph{Content Store}, which can store copies of content objects.  All Content Stores are assumed to have the same size: $K$ content units. 
	
	Time is slotted: $t = 0,1,2,\ldots$.  Assuming an infinite backlog of requests at each node, all nodes generate requests for content objects at each time $t$.  Each content request is for content object $m, 1 \leq m \leq M$, with probability $p_m$, independent of all other requests.  Content requests are admitted into the network at the rate of the achievable throughput for a feasible scheme.

{Since the content popularity distribution is assumed to be time-invariant, we implement a static caching allocation in the initial phase of the network operation.   Let $\chi_m$ be the set of nodes which cache content object $m$ in their Content Store, where $X_m=|\chi_m|$.
We call the nodes in $\chi_m$ the \emph{holders} of content $m$.  The holders are specifically chosen as follows.
For each content $m$, choose one of the  $\binom{n}{X_m}$ sets of $X_m$ nodes, uniformly at random and independent of the set choices for all other contents, and designate the nodes in the chosen set as the holders of content $m$.  This ensures that for each $m$, there are exactly $X_m$ holders distributed uniformly and independently in the network. In addition, the sets of holders are chosen independently across different contents.}
	
In order for a caching allocation $\{X_m\}_{m=1}^M$ to be feasible, the constraint on total caching space must be satisfied:
	\begin{equation}
	\sum_{m=1}^M X_m \leq nK.
	\label{eq:totalcache}
	\end{equation}
	The total caching constraint in~\eqref{eq:totalcache} is a relaxed version of the individual caching constraints.  For ease of presentation and analysis, we use~\eqref{eq:totalcache} for the throughput-delay analysis and optimization problem. 	
	
	For concreteness, we consider the content delivery mechanism embodied in the NDN architecture \cite{ndn}.  Specifically, requests for content objects are submitted using Interest Packets, which are forwarded toward Content Stores caching the requested content object using multi-hop communication.\footnote{Assume that routing (topology discovery and data reachability) has already been accomplished in the network, so that each node knows to which other nodes it can forward an Interest Packet to reach a Content Store caching the requested object.  Equivalently, in an NDN network, the Forward Information Base (FIB) has already been populated at each node for each content object.}
	When the Interest Packet reaches a node caching the requested content object, a Data Packet containing the requested content object is transmitted in the reverse direction along the path taken by the corresponding Interest Packet, back to the requesting node.\footnote{Note that Interest Packets are usually much smaller in size than the corresponding Data Packet.  If a node requests a content object which is cached in its local Content Store, the request can be satisfied immediately and there is no need to generate an Interest Packet.  Since the Content Store has limited cache space, this is not usually the case.  For ease of analysis, we assume in this paper that if the requested content is in the local cache, the node still generates an Interest Packet for it, transmits it to the nearest holder excluding itself, and uses the network to retrieve the content object.} 
%In this paper, we assume that an Interest Packet requesting object $m$ is forwarded along the direct line connecting the requesting node to the {\em closest} (in Euclidean distance) holder of content object $m$, using multi-hop communication.

	Transmissions in wireless networks are subject to multi-user interference.  Our model for a successful wireless transmission in this environment follows the {\em Protocol Model} given in~\cite{elgamal}. 
	Suppose node $i$ transmits a packet at time $t$. Then, a node $j$ can receive this packet successfully if and only if for any other  node $k$ transmitting simultaneously, {$|U_k -U_j | \geq (1 + \Delta)|U_i - U_j|$}, where $U_i$ is the location of node $i$, $|\cdot|$ denotes Euclidean distance, and $\Delta$ is a positive constant.
	During a successful transmission, the transmitter sends at a rate of $W$ bits per second, which is a constant independent of $n$.  
	Another model for transmission is the Physical
	Model~\cite{gupta}. Since these two models are essentially equivalent
	(assuming a path loss exponent of greater than 1 and equal node transmission powers in the Physical Model)~\cite{gupta}, we focus on the Protocol
	Model in this paper.

	To simplify our analysis, we adopt the {\em fluid model} for packet transmission considered in~\cite{elgamal}.  In the fluid model, we allow the size of the content unit, and therefore the sizes of the Interest Packets and Data Packets,  to be arbitrarily small, depending on the number of nodes in the network. Thus, the time required for transmitting an Interest Packet or Data Packet is much smaller than a time slot.  Nevertheless, a packet received by a node in a given time slot cannot be transmitted by the node until the next time slot.  Thus, all packets waiting for transmission at a given node will be transmitted by the node in one time slot.   The fluid model makes unnecessary detailed analysis of the scheduling of individual packets.  As explained below, we will specifically assume that the packet size scales in proportion to the per-node throughput of the achievable scheme.
	
	\section{Throughput and Delay}
	\label{comm}
	Transmission and caching in the wireless network are coordinated
	and controlled by a {\em scheme}. More precisely,
	a scheme $\pi$ is a sequence of policies $\{\pi_n\}$,
	where $\pi_n$ determines the (static) caching allocation, as well as
	the scheduling of transmissions in each time slot, for a network of $n$ nodes.
	For a given scheme, the throughput and delay are defined as follows:
	
	\begin{defi}[Throughput]
		\label{def2}
		For a given scheme $\pi_n$, let $B_{\pi_n}(i,t)$ be the total number of bits of all content objects received by the requesting node $i$ up to time $t$. 
		The long-term throughput of node $i$ is
		$$ \liminf_{t \to \infty}\frac{1}{t}B_{\pi_n}(i,t). $$
		The average throughput over all nodes is
		\begin{eqnarray*}
			\lambda'_{\pi_n}(n) & = & \frac{1}{n}\sum_{i=1}^{n}\liminf_{t \to \infty}\frac{1}{t}B_{\pi_n}(i,t).
		\end{eqnarray*}
		The throughput of $\pi_n$, is defined as the expectation over
		all realizations of node positions $\{U_1, U_2, \ldots, U_n\}$, of the corresponding average throughput:
		$$  \lambda_{\pi_n}(n)\triangleq E\left[ \lambda'_{\pi_n}(n) \right].
		$$
	\end{defi}
	
	\begin{defi}[Delay]
		For a given $\pi_n$, let $D_{\pi_n}(i,k)$ be the delay of the $k$-th request for any content object by node $i$ (measured from the moment the Interest Packet leaves $i$ for the closest holder until the corresponding Data Packet arrives at $i$ from the holder).  The delay (over all content requests) for node $i$ is $$\limsup_{r \to \infty}\frac{1}{r}\sum_{k=1}^{r}D_{\pi_n}(i,k).$$ 
		The average delay over all nodes is
		\begin{eqnarray*}
			D'_{\pi_n}(n) & = & \frac{1}{n}\sum_{i=1}^{n}\limsup_{r \to \infty}\frac{1}{r}\sum_{k=1}^{r}D_{\pi_n}(i,k).
		\end{eqnarray*}
		The delay of $\pi_n$ is defined as the expectation over
		all realizations of node positions $\{U_1, U_2, \ldots, U_n\}$, of the corresponding average delay:
		\begin{equation*}
		D_{\pi_n}(n)\triangleq E\left[ D'_{\pi_n}(n) \right].
		\end{equation*}
	\end{defi}
	
	The throughput and delay quantities $\lambda'_{\pi_n}(n)$ and $D'_{\pi_n}(n)$
	are random variables, since they depend on the realization of node positions. 
	The quantities  $\lambda_{\pi_n}(n)$ and $ D_{\pi_n}(n)$ are ensemble averages.   Note that due to the stationarity and ergodicity of the content request sequences, the throughput and delay quantities in Definitions 1 and 2 are well defined.  That is, the random content request sequences are averaged over in the throughput and delay definitions.  
	 To study the asymptotical
	behavior of $\lambda_{\pi_n}(n)$ and $ D_{\pi_n}(n)$,
	we will let the number of nodes
	$n$ go to infinity. 
	%We say an event holds with high probability (w.h.p.) if the event occurs with probability 1 as $n$ goes to infinity.
	
	Recall from Section II that for each $m$, there are $X_m$ holders distributed uniformly and independently in the network area.  Furthermore, the sets of holders are chosen independently across different contents.  To analyze the throughput and delay scaling of the content-centric wireless network, we combine this caching allocation scheme with an achievable multi-hop, TDM, cellular transmission scheme~\cite{elgamal}.  In this scheme, the unit torus is divided into square cells, each with area $a(n)$.\footnote{We ignore the imperfection of the square cells as well as edge effects due to $1/a(n)$ not being a perfect square.}   We use the following sequence of lemmas to construct the transmission and caching scheme yielding the main throughput and delay scaling result.  
	
	The following lemma from~\cite{elgamal} shows that with an appropriately chosen cell area $a(n)$, each cell has at least one node w.h.p., so that multi-hop relaying of packets through adjacent cells is possible.
	\begin{lemma}\cite{elgamal}
		\label{lemma1}
		If $a(n)\geq 2 \log n/n$, then each cell has at least one node w.h.p..
	\end{lemma}
	For $a(n)$ satisfying Lemma~\ref{lemma1}, we set the transmission radius to be $r(n) = \sqrt{8 a(n)}$.  This allows each node to transmit to nodes within its cell and to the 8 neighboring cells.  It is then clear that multi-hop packet relaying through adjacent cells can take place w.h.p.  
	
	The next lemma from~\cite{elgamal} makes possible the establishment of an interference-free TDM transmission schedule where each cell becomes active (i.e. any of the nodes in the cell transmits) regularly once every $N+1$ time slots, where $N$ is specified in Lemma~\ref{lemmaprotocol}, and no two simultaneously active cells interfere with each other.  Here, two simultaneously active cells interfere if the transmission of a node in one active cell affects the success of a simultaneous transmission by a node in the other active cell.
	\begin{lemma}\cite{elgamal}
		\label{lemmaprotocol}
		Under the Protocol model, the number of cells that interfere with any given cell is bounded above by a constant $N=16(1+\Delta)^2$, independent of $n$.
	\end{lemma}

	We consider a transmission scheme where an Interest Packet requesting content object $m$ is forwarded along the direct line connecting the requesting node to the {\em closest} (in Euclidean distance) holder of content object $m$, using multi-hop communication.   The next lemma computes the expected Euclidean distance from a given node requesting content $m$ to the closest holder of content $m$. 
	\begin{lemma}
		\label{lemma2}
		Let $\chi_m$ be the set of holders of content $m$, independently and uniformly distributed in the unit-sized network area, where $X_m=|\chi_m|$. For any node requesting content $m$, the average Euclidean distance from the requesting node to the closest holder of content $m$ is $\Theta(\frac{1}{\sqrt{X_m}})$.
		%\footnote{whether the requester is one of the holders or not, this Lemma  still holds true, noting that in the case where the requester is one of the holders, the distance to the closest holder excluding the requester itself is of order $\Theta(\frac{1}{\sqrt{X_m-1}}) =\Theta(\frac{1}{\sqrt{X_m}})$.}. 
	\end{lemma}
	\begin{proof}
		Please see Appendix \ref{a}.
	\end{proof}
	
\vspace{0.1in}
	Assume   $a(n) \geq 2 \log n/n$ and $r(n) = \sqrt{8 a(n)} \geq 4\sqrt{ \log n/n}$. Consider a fixed node $i$ requesting content object $m$. Let $L_{H,R}(i,m)$ be the straight line connecting $i$ to the closest holder of content $m$. From Lemma \ref{lemma2}, 
	\begin{equation}
	\label{length}
	E\left[|L_{H,R}(i,m)|\right]=\Theta\left(\frac{1}{\sqrt{X_m}}\right).
	\end{equation}
	where $|L|$ denotes the Euclidean length of line $L$. Let $H_{i,m}$ be the number of hops along a path (sequence of nodes) which originates at 	
	requester $i$ and ends at the closest holder of content $m$, and lies within the set of cells intersecting the $L_{H,R}(i,m)$ line, where there is {\em  exactly one node per cell} along the path.

By Lemma \ref{lemma1}, we can find at least one node per cell w.h.p.  Therefore, we can construct the described path w.h.p.

Note that since we are requiring the path to have exactly one node per cell, the path is not necessarily the shortest path (in terms of the number of hops) connecting requester $i$ and the closest holder of content $m$, which lies within the set of cells intersecting the $L_{H,R}(i,m)$ line.
On the other hand, we show in the following lemma that the expected value of $H_{i,m}$ is of the same order as the expected value of $H'_{i,m}$, where $H'_{i,m}$ is the minimum number of hops along the shortest path.

\color{black}
	\begin{lemma}
		\label{lemma_hops}
		For $a(n)\geq 2\log n/n$, and each $m = 1, \ldots, M$,
		\begin{equation}
		 E[H_{i,m}]=\Theta\left(E[H'_{i,m}]\right)= \Theta \left(\max{\left\lbrace\frac{1}{\sqrt{a(n)X_{m}}},1\right\rbrace}\right) \; w.h.p.  \label{formhops}
		\end{equation}
	\end{lemma}
	\begin{proof}
		Please see Appendix \ref{app_hops}.
	\end{proof}
	
	\vspace{0.1in}
	We now prove a key lemma, characterizing the number of $L_{H,R}(i,m)$ lines passing through each cell as $n$ becomes large.  The result
	may be seen as an analogue of Lemma 3 in~\cite{elgamal} for the wireless caching network environment.
	\begin{lemma}
		\label{lemma3}
		For $a(n)\geq 2 \log n/n$, the number of $L_{H,R}$ lines passing through each cell is $$\Theta\left(n\sum_{m=1}^{M}p_m \max{\{\sqrt{a(n)/X_m},a(n)\}}\right) \; w.h.p.$$  
	\end{lemma}
	\begin{proof}
		For a given content request vector $(m_1,m_2,\ldots,m_n)$ at time $t$ and a given node $i$, we know that $H_{i,m_i}=H_{i,m}$, w.p. $p_m$, for $m=1,2,\ldots,M$. Therefore, 
		\begin{IEEEeqnarray}{+rCl+x*}
			\label{EHMi}
			E[H_{i,m_i}]&=& \sum_{m=1}^{M}p_m E[H_{i,m}]\nonumber\\
			&=& \Theta\left(\sum_{m=1}^{M}p_m\max{\left\lbrace\frac{1}{\sqrt{a(n)X_{m}}},1\right\rbrace}\right). \label{EH}
		\end{IEEEeqnarray}
		There are $1/a(n)$ cells. Fix a cell $j$ and let $Y_{i,m_i}^j$ be the indicator of the event that the $L_{H,R}(i,m_i)$ line passes through cell $j$. That is,
		\begin{displaymath}
			Y_{i,m_i}^j =
			\begin{cases}
				1, & \text{if } L_{H,R}(i,m_i) \text{ passes through cell } j \\
				0, & \text{otherwise}\\
			\end{cases}
		\end{displaymath}
		for $1\leq i \leq n$, $1\leq j \leq 1/a(n)$ and $1\leq m_i\leq M$. We know that $Y_{i,m_i}^j=Y_{i,m}^j$, w.p. $p_m$, for $m=1,2,\ldots,M$. Hence, we obtain $E[Y_{i,m_i}^j]= \sum_{m=1}^{M}p_m E[Y_{i,m}^j]$. Summing up the total number of hops for any $m$ in two different ways gives us:
		\begin{equation}
			\label{summation}
			\sum_{i=1}^{n}\sum_{j=1}^{1/a(n)}Y_{i,m}^j =\sum_{i=1}^nH_{i,m}.
		\end{equation}
		Taking the expectation on the both sides of (\ref{summation}), and noting that $E[H_{i,m}]$ is the same for each node $i$ and $E[Y_{i,m}^j]$ is equal for every $i$ and $j$ due to symmetry of the torus, we have
		\begin{equation}
			\sum_{i=1}^{n}\sum_{j=1}^{1/a(n)}E\left[ Y_{i,m}^j\right]=\sum_{i=1}^n E\left[ H_{i,m}\right]. \nonumber 
		\end{equation}
		\begin{equation}
			nE[Y_{i,m}^j]/a(n)=nE[H_{i,m}]. \nonumber
		\end{equation}
		Therefore,
		\begin{IEEEeqnarray}{+rCl+x*}
			E[Y_{i,m}^j]&=& {a(n)}\cdot E\left[H_{i,m}\right] \nonumber\\
			& = & \Theta\left(  \max{\{\sqrt{a(n)/X_m},a(n)\}}\right). \IEEEeqnarraynumspace
		\end{IEEEeqnarray}
		Now, 
		\begin{IEEEeqnarray}{+rCl+x*}
			E[Y_{i,m_i}^j]&=& \sum_{m=1}^{M}p_m E[Y_{i,m}^j]\nonumber\\
			&=& \Theta\left(\sum_{m=1}^{M}p_m\max{\lbrace\sqrt{a(n)/X_{m}},a(n)\rbrace}\right).
		\end{IEEEeqnarray}
		The total number of $L_{H,R}$ lines passing through a fixed cell $j$, is given by $Y=\sum_{i=1}^n Y_{i,m_i}^j$. Hence, $E[Y]= \Theta(n\sum_{m=1}^M p_m \max{\lbrace\sqrt{a(n)/X_{m}},a(n)\rbrace} )$.   Recall that nodes are independently and uniformly distributed in the unit-sized network area and requesters request contents independently from one another.  Moreover, across different contents, the sets of holders are chosen independently.  Therefore, it can be shown that for each cell $j$, $(Y_{i,m_i}^j)_{i=1,\cdots,n}$ is a set of independent random variables satisfying
$0\leq Y_{i,m_i}^j\leq 1$. \color{black} Applying the Chernoff bound yields \cite{chernoff}
		\begin{equation}
			\label{delta}
			P\{Y>(1+\delta)E[Y]\}\leq \exp\left(-\frac{\delta^2E[Y]}{3}\right).
		\end{equation}
		Choosing $\delta=\sqrt{6\log n/E[Y]}$, we are guaranteed that $\delta=o(1)$. This is true as we are assuming that $a(n)=\Omega(\log n/n)$. Also, as explained later, there is no need for any content object to have more than $\Theta(1/a(n))$ holders. Due to the total caching capacity constraint, $\sum_{m=1}^{M}X_m\leq nK$, and the fact that $M=\Theta(n^{\beta})$, where $0<\beta<1$, we are assured that $E[Y]=\omega(na(n))$, or equivalently, $E[Y]=\omega(\log n)$, resulting in $\delta=o(1)$. Substituting $\delta$ in (\ref{delta}), we have
		\begin{equation}
			P\{Y>(1+\delta)E[Y]\}\leq 1/n^2.
		\end{equation}
		Therefore, $Y=O(E[Y])$ with probability $\geq 1-1/n^2$. Similarly, by applying the Chernoff bound to the lower tail~\cite{chernoff}, we have
		\begin{equation}
			\label{delta2}
			P\{Y<(1-\delta)E[Y]\}\leq \exp\left(-\frac{\delta^2E[Y]}{2}\right).
		\end{equation}
		Applying similar techniques as above, we can show that $Y=\Omega(E[Y])$ with probability $\geq 1-1/n^2$.  Now applying the union bound over all $8/r^2(n)$ cells, we see that the number of  $L_{H,R}$ lines passing through each cell of the network is $$\Theta(E[Y])= \Theta\left(n\sum_{m=1}^M p_m \max{\lbrace\sqrt{a(n)/X_{m}},a(n)\rbrace}\right).$$ with probability $\geq 1-1/n$.
	\end{proof}

%	The above lemma computes (w.h.p.) the order of $Y$, the average total number of $L_{H,R}$ lines passing through each cell. 
	
	We now present in detail the achievable caching and transmission scheme.  The transmission scheme can be seen as an analogue of Scheme 1 in  \cite{elgamal}, for the wireless caching network environment.  The scheme is parameterized by the cell area $a(n)$, where $a(n) = \Omega(\log n/n)$ and $a(n) \leq 1$. 
	
	\subsection{Caching Scheme}
	
	For each content $m$, choose one of the  $\binom{n}{X_m}$ sets of $X_m$ nodes, uniformly at random and independent of the set choices for all other contents, and designate the nodes in the chosen set as the holders of content $m$.  This ensures that for each $m$, there are exactly $X_m$ holders distributed uniformly and independently in the network. In addition, the sets of holders are chosen independently across different contents.

\subsection{Transmission Scheme} 
\begin{enumerate}
	\item Divide the unit torus using a square grid into square cells, each with area $a(n)$.
	\item For the given realization of the random network, check that there is no empty cell.  
	%the following conditions are satisfied:
	%\begin{itemize}
	%	\item Condition 1: No cell is empty.
	%	\item Condition 2: The number of $L_{H,R}$ lines passing through each cell satisfies (\ref{cond2-1}) and (\ref{cond2-2}).
	%\end{itemize}
	%	\item If either of the above conditions is not satisfied, 
	\item If there is an empty cell, then use a time-division policy, where each of the $n$ requesters communicates directly with the closest holder of the requested content object, in a round-robin fashion.
	%	\item If both conditions are satisfied
	\item Otherwise, use the following policy $\pi_n$:
	\begin{enumerate}
		\item Each cell becomes active regularly once every $1+N$ time-slots (Lemma \ref{lemmaprotocol}).  Cells which are sufficiently far apart become active simultaneously.   That is, the scheme uses TDM between neighboring cells.
		\item Requesting nodes transmit Interest Packets to the closest holders by hops along the adjacent cells intersecting the 
		$L_{H,R}$ lines. Similarly, the holders transmit Data Packets to the requesting nodes along the same path taken by their corresponding Interest Packets, in the reverse direction.
		\item  Each time slot is split into two sub-slots. In the first sub-slot, each active cell transmits a single Interest Packet for each of the $L_{H,R}$ lines passing through the cell toward the closest holder. In the second sub-slot, the active cell transmits a single Data Packet for each of the $L_{H,R}$ lines passing through the cell toward the requesting node.
	\end{enumerate}  
\end{enumerate}

	We now derive the throughput and delay performance of the achievable transmission and caching scheme described above, for a given feasible caching allocation $\{X_m\}_{m=1}^M$.
	  We further show that the achievable transmission/caching scheme attains the order-optimal throughput and delay performance, among all  transmission/caching schemes where for each $m$, the $X_m$ holders are independently and uniformly distributed in the network area, and each node has the same transmission radius $r(n) = \sqrt{8a(n)}$. As explained in Section \ref{optimizationsection}, we then optimize the delay and throughput of the achievable scheme simultaneously by selecting optimal $(X_m)_{m=1}^M$ subject to caching constraints. 
	\color{black}
	
	\begin{theorem}
		\label{throughputlemma}
		For $a(n) \geq 2 \log n /n$, the throughput and delay scaling of the achievable caching and transmission scheme are given by
		\begin{equation}
		\label{lowerbound}
		\lambda(n)= \Theta \left(\frac{1}{n\sum_{m=1}^M p_m \max{\lbrace\sqrt{a(n)/X_{m}},a(n)\rbrace}} \right) \; w.h.p.
		\end{equation}
		\begin{equation}
		\label{delaybound}
		D(n) =\Theta\left(\sum_{m=1}^{M}p_m \max{\left\lbrace\frac{1}{\sqrt{a(n)X_m}},1\right\rbrace}\right) \; w.h.p.
		\end{equation}
		
		 Furthermore, the achievable transmission/caching scheme attains the order-optimal throughput and delay performance, among all  transmission/caching schemes where for each $m$, the $X_m$ holders are independently and uniformly distributed in the network area, and each node has the same transmission radius $r(n) = \sqrt{8a(n)}$.
		%\begin{equation}
		%\label{tradeoff}
		%D(n)\lambda(n)= \Theta((na(n))^{-1} ) \; w.h.p.
		%\end{equation}
	\end{theorem}
	\begin{proof}
		\label{app_throughputlemma}
		First note that if the time-division policy with direct communication is used, then the throughput is $W/n$ with a delay of 1. But since this happens with a vanishingly low probability, as shown by Lemma \ref{lemma1}, the throughput and
		delay for the achievable scheme are determined by that of policy $\pi_n$.
		When policy $\pi_n$ is used, each cell
		has at least one node. This assures us that requester-holder pairs can communicate with each other by hops along adjacent
		cells on their $L_{H,R}$ lines. From Lemma \ref{lemmaprotocol}, each cell gets to transmit packets every $1+ N$ time-slots. Hence, the cell throughput is $\Theta(1)$. The total traffic through each cell is due to all the $L_{H,R}$ lines passing through the cell, which is  $\Theta(n\sum_{m=1}^{M}p_m \max{\{\sqrt{a(n)/X_m},a(n)\}})$ w.h.p. This shows that 
		\begin{equation}
		\lambda(n)= \Theta \left(\frac{1}{n\sum_{m=1}^M p_m \max{\lbrace\sqrt{a(n)/X_{m}},a(n)\rbrace}} \right) w.h.p.
		\end{equation}
		\color{black}
		{Substituting $a(n) = r^2(n)/8$, it follows that
			\begin{equation}\label{lowerbound_r}
			\lambda(n)= \Theta \left(\frac{1}{n\sum_{m=1}^M p_m \max{\lbrace r(n)/\sqrt{X_{m}},r^2(n)\rbrace}} \right) w.h.p.
			\end{equation}}
		
		Recall that by Lemma \ref{lemmaprotocol}, each cell can be active once every $N+1$ time-slots, where $N$ is constant and independent of $n$.  As we are assuming that packets scales in proportion to the throughput $\lambda(n)$ (fluid model), each packet arriving at a node in the cell departs in the
		next active time-slot of the cell. Hence, the packet delay is  $N+1$ times the number of hops from the requester to the holder. For a given realization of the random network, where node $i$ is requesting $m_i$ for $i=1,2,\ldots,n$, and $m_i \in \{1,2,\ldots,M\}$, let $h_{i,m_i}$ be the number of hops from the requester $i$ to its closest holder of content $m_i$ in the given realization. Furthermore, since the Data Packet takes  the same path as the corresponding Interest Packet in reverse, the average delay of the network realization is given by two times the mean sample of the $h_{i,m_i}$'s, i.e. $\frac{2}{n}\sum_{i=1}^{n}h_{i,m_i}$. As $n \rightarrow \infty$, by the Law of Large Numbers, 
		\begin{equation}
		\frac{2}{n}\sum_{i=1}^{n}h_{i,m_i}\simeq 2E[H_{i,m_i}].
		\end{equation}
		Using (\ref{EH}), equation (\ref{delaybound}) follows.

		 Now consider any transmission/caching scheme where for each $m$, the $X_m$ holders are independently and uniformly 						distributed in the network area, and each node has the same transmission radius $r(n) = \sqrt{8a(n)}$.  We show that the throughput
		 and delay performance of such a scheme cannot be strictly better than~\eqref{lowerbound}-\eqref{delaybound} in an order sense.
		 
		 By Theorem 5.13 in \cite{gupta}, the common transmission radius must satisfy $r(n) = \Omega(\sqrt{\log n/n})$ in order to have 	      		 no isolated node in the network w.h.p.  Next, it is shown in \cite{gupta} that under the Protocol Model, the maximum number of 				simultaneous transmissions 	feasible in a dense random network is no more than
		\begin{equation}
		\frac{1}{\frac{1}{4\pi}\cdot\frac{\pi\Delta^2 r^2(n)}{4}}=\frac{16}{\Delta^2 r^2(n)}.
		\end{equation}
		This is due to the fact that each transmission consumes an area of radius $\frac{\Delta}{2}r(n)$ around every transmitter, and at least \smash{$\frac{1}{4\pi}$} portion is within the unit torus. 
		
%		Note that in order to have a stable network, the arrival rate of requests for contents must be less than the achievable throughput. 

	Note that since each node transmits with radius $r(n)$, it follows from Lemma~4 that the minimum number of hops that an Interest Packet requesting content $m_i$ travels from requester $i$ to reach the closest holder is $H'_{i,m_i}$.  Due to symmetry on the torus, the bits per second being transmitted simultaneously by the whole network for all the contents must be at least $n\lambda(n)E[H'_{i,m_i}]$, where $\lambda(n)$ is the per-node throughput. Therefore, we have 
		\begin{equation}
		\label{constraint}
		n\lambda(n)E[H'_{i,m_i}]\leq \frac{W}{1+c} \cdot \frac{16}{\Delta^2 r^2(n)}.
		\end{equation}
		where $0<c \leq 1$ is the ratio of the Interest Packet size to the corresponding Data Packet size. 
		Since $H'_{i,m_i} = H'_{i,m}$ w.p. $p_m$, an upper bound on the per-node throughput is obtained:
		\begin{IEEEeqnarray}{+rCl+x*}
			\label{upperineq}
			\lambda (n) &\leq& \frac{W}{1+c} \cdot \frac{16}{\Delta^2 r^2(n)n\sum_{m=1}^{M}p_m E[H'_{i,m}]}.
		\end{IEEEeqnarray}
		By Lemma \ref{lemma_hops}, it follows that
		\begin{equation}
		\label{upperbound_d}
		\lambda (n)=O \left(\frac{1}{n\sum_{m=1}^M p_m \max{\{\frac{r(n)}{\sqrt{X_m}},r^2(n) \}}} \right),
		\end{equation}
		thus showing that the throughput attained by the achievable scheme in (\ref{lowerbound_r}) is order-optimal. 
		
		%The inequality given in (\ref{upperineq}) shows that the optimal strategy needs to use the shortest path algorithm toward the destination. Our proposed scheme does not take the shortest path for communication, however, as shown in  Lemma \ref{lemma_hops}, the average number of hops which is required in the proposed scheme is of the same order of the shortest path algorithm. Hence, the proposed scheme is of the same order of the optimal solution.
		
		Now for the network delay: under the fluid model, the average delay is simply $2(N+1)$ times the number of hops.  Thus, by Lemma 4 and by symmetry, the average delay is lower bounded by $E[H'_{i,m_i}]$, which by Lemma 4, is equal in order to $E[H_{i,m_i}]$.  Thus, the delay attained by the achievable scheme in~\eqref{delaybound} is order-optimal.
				\color{black}
	\end{proof}
	
%	\begin{theorem}
%	The network delay and per-node throughput of the proposed network under the described transmission and caching scheme, given in (\ref{lowerbound}) and (\ref{delaybound}), are optimal in the order sense.
%	\label{theorem}
%%	\begin{equation}
%%	\label{throughputbound}
%%	\lambda(n)= O \left(\frac{1}{n\sum_{m=1}^M p_m \max{\{\sqrt{\frac{a(n)}{{X_m}}},a(n) \}}} \right).
%%	\end{equation}
%%	\begin{equation}
%%		\label{delayOptimalbound}
%%		D(n) =\Omega\left(\sum_{m=1}^{M}p_m \max{\left\lbrace\frac{1}{\sqrt{a(n)X_m}},1\right\rbrace}\right).
%%	\end{equation}
%	\end{theorem}

	Note that the per-node throughput and network delay given in Theorem 1 satisfy the following relation:
		\begin{equation}
			\label{tradeoff}
			D(n)\lambda(n)= \Theta((na(n))^{-1} ) \; w.h.p.,
			\end{equation} 
This holds for any feasible caching allocation set $(X_m)_{m=1}^M$. Equation \eqref{tradeoff} states that for a given $n$ and $a(n)$, maximizing throughput is equivalent to minimizing the network delay. In the next section, we find the optimized set $(X_m)_{m=1}^M$ which minimizes the delay, or equivalently maximizes the throughput. 
\color{black}	
	\section{Optimized Caching}
	\label{optimizationsection}
	
	We now optimize the delay and throughput of the achievable transmission and caching scheme described in Section~\ref{comm}, by selecting the appropriate $(X_m)_{m=1}^M$ subject to caching constraints.  We first relax the integer constraint on $(X_m)_{m=1}^M$, thus allowing $X_m$ to be a non-negative real number.\footnote{It can easily be shown that the integer constraint relaxation does not change the order of the optimal delay and throughput scaling.}  Furthermore, we enforce only the total caching constraint in~\eqref{eq:totalcache}, which is a relaxation of the per node caching constraint. 
	
	To illustrate the optimization process, we focus on the commonly used Zipf distribution as the content popularity distribution~\cite{asymptotic, italian}.  Let $p_m= m^{-\alpha}/H_{\alpha}(M)$, where $\alpha$ is the Zipf's law exponent, and $H_{\alpha}(M)=\sum_{i=1}^M i^{-\alpha}$ is a normalization constant, given by \cite{asymptotic} 
	\begin{equation}
	\label{H}
	H_{\alpha}(M)=
	\begin{cases}
	\Theta(1), & \alpha>1\\
	\Theta \left(\log M\right), & \alpha=1\\
	\Theta(M^{1-\alpha}), & \alpha<1\\
	\end{cases}
	\end{equation}
	
	As can be seen, for the case $X_m = \Omega (a^{-1}(n))$, $\lambda(n)$ and $D(n)$ are independent of the number of holders. Hence, there is no need to cache more than one copy of any given content object in any one cell.
	Also, note that by (\ref{tradeoff}), minimizing the delay is equivalent to maximizing the throughput.  We may obtain the minimum delay by solving the following optimization problem:
	\begin{equation}
	\begin{cases}
	\ \min_{\{X_m\}} \sum_{m=1}^M \frac{p_m}{\sqrt{a(n)X_m}} \\
	\ \text{subject to:}\\
	\ \sum_{m=1}^M X_m \leq nK \\
	\ 1 \leq X_m \leq a^{-1}(n)  \qquad \text{for } m=1,2,\ldots,M\\
	\end{cases}
	\label{optimization}
	\end{equation}
	
	As the objective function is strictly convex, we are assured that there is a unique global minimum. Defining the non-negative Lagrange multipliers $\lambda$ for the constraint $\sum_{m=1}^{M}X_m\leq nK$, and taking into account the  constraint $1 \leq X_m \leq a^{-1}(n)$, the necessary conditions for a minimum of $D$ with respect to $X_m$, $\forall m \in M$ are given by
	\begin{equation}
	\frac{\partial D}{\partial X_m}
	\begin{cases}
	\ \leq -\lambda \qquad \text{if } X_m=a^{-1}(n) \\
	\ = -\lambda \qquad \text{if } 1<X_m<a^{-1}(n)\\
	\ \geq -\lambda \qquad \text{if } X_m=1 \\
	
	\end{cases}
	\label{necessary1}
	\end{equation}
	
	For the Zipf distribution, it is clear that $p_m$ is strictly decreasing in $m$ and therefore so is $X_m$. 
	Hence, let $\mathcal{M}_1=\{1,2,\ldots,m_1-1\}$ be the set of content objects such that $X_m=a^{-1}(n)$ for $m\in\mathcal{M}_1$. Similarly, let $\mathcal{M}_2=\{m_1,m_1+1,\ldots,m_2-1\}$ and $\mathcal{M}_3=\{m_2,m_2+1,\ldots,M\}$ be the set of contents such that $1<X_m<a^{-1}(n)$ for $m\in\mathcal{M}_2$, and $X_m=1$ for $m\in\mathcal{M}_3$, respectively. From (\ref{necessary1}), we have $\forall m \in M$
	\begin{equation}
	\frac{p_m}{2\sqrt{a(n)X_m^3}}
	\begin{cases}
	\ \geq \lambda \qquad  \forall m\in\mathcal{M}_1\\
	\ = \lambda \qquad \forall m\in\mathcal{M}_2\\
	\ \leq \lambda \qquad \forall m\in\mathcal{M}_3 \\
	
	\end{cases}
	\label{necessary2}
	\end{equation}
	
	Using the equality for the case $m\in\mathcal{M}_2$, we obtain
	\begin{equation}
	\label{m1/m2}
	\frac{m_1}{m_2}\simeq (a(n))^{\frac{3}{2\alpha}}.
	\end{equation}
	
	Clearly from (\ref{necessary2}), we have $\lambda>0$ and hence, $\sum_{m=1}^{M}X_m=nK$. Combining this with (\ref{m1/m2}), we can derive $m_1$ and $m_2$. The optimal number of holders of content $m$, $X_m^*$, is then given by 
	\begin{equation}
	\label{Xm}
	X_{m}^*=
	\begin{cases}
	\  a^{-1}(n), \quad\qquad\quad m=1,2,\ldots, m_1-1\\
	\ \frac{p_m^{2/3}}{\sum_{j=m_1}^{m_2-1}p_j^{2/3}} nK', m=m_1,\ldots,m_2-1\\
	\ \quad 1, \quad\qquad\qquad m=m_2,\ldots,M
	\end{cases}
	\end{equation}
	where $K'\triangleq K-{(m_1-1)}\frac{a^{-1}(n)}{n}-\frac{(M-m_2+1)}{n}$.  The average delay is then w.h.p.:
	\begin{equation}
	\begin{split}
	D^*(n)=\Theta \left( \sum_{j=1}^{m_1-1}p_j +\right. \frac{\left(\sum_{j=m_1}^{m_2-1}p_j^{2/3}\right)^{3/2}}{\sqrt{nK'a(n)}}\\
	\left.+  \frac{\sum_{j=m_2}^{M}p_j}{\sqrt{a(n)} } \right).
	\end{split}
	\label{D}
	\end{equation}
	
	To gain more insight on the structure of the optimal solution, we have the following lemma.
	
	\begin{lemma}
		\label{m1m2lemma}
		As $n\rightarrow \infty$, the scaling of indices $m_1$ and $m_2$ is given by
		\begin{equation}
		\label{m1}
		m_1=
		\begin{cases}
		\ \Theta(\min\{M,na(n)\}), & \alpha>3/2\\
		\ \Theta\left(\min\{M,\frac{na(n)}{\log n}\}\right), & \alpha=3/2\\
		\ \Theta\left(\max\{1,\min\{M,na(n),\frac{(na(n))^{\frac{3}{2\alpha}}}{M^{\frac{3}{2\alpha} -1}}\}\}\right), & \alpha<3/2
		\end{cases}
		\end{equation}
		\begin{equation}
		\label{m2}
		m_2=
		\begin{cases}
		\ \min\{M+1,\frac{2\alpha-3}{2\alpha}nK(a(n))^{1-\frac{3}{2\alpha}}\}, & \alpha>3/2\\
		\ M+1, & \alpha\leq 3/2
		\end{cases}
		\end{equation}
	\end{lemma}
	\begin{proof}
		Refer to Appendix \ref{app_m1m2lemma}
	\end{proof}
	We can now compute the optimized delay and throughput for the achievable scheme, assuming
	$M=\Theta(n^\beta)$ where $0<\beta<1$, under the Zipf popularity distribution.
	
	\begin{theorem}
		\label{theorem5}
		For $a(n) \geq 2 \log n /n$, the throughput and delay of the proposed scheme using Zipf distribution are  w.h.p.:
		\begin{equation}
		\label{delayadhoc}
		D^*(n)=
		\begin{cases}
		\Theta(1), & \alpha>3/2\\
		\Theta\left(\max\{1,\frac{(\log M)^{3/2}}{\sqrt{na(n)}}\}\right), & \alpha=3/2\\
		\Theta \left(\max\{1,\frac{M^{3/2-\alpha}}{\sqrt{na(n)}}\}\right), & 1<\alpha<3/2\\
		\Theta \left(\max\{1,\frac{\sqrt{M}}{\log M\sqrt{na(n)}}\}\right), & \alpha=1\\
		\Theta  \left(\max\{1,\sqrt{\frac{M}{na(n)}}\}\right), & \alpha<1\\
		\end{cases}
		\end{equation}
		\begin{equation}
		\label{throughputadhoc}
		\lambda^*(n)=
		\begin{cases}
		\Theta(\frac{1}{na(n)}), & \alpha>3/2\\
		\Theta  \left(\max\{\frac{1}{n},\frac{1}{(\log M)^{3/2}\sqrt{na(n)}}\}\right), & \alpha=3/2\\
		\Theta \left(\max\{\frac{1}{n},\frac{M^{\alpha-3/2}}{\sqrt{na(n)}}\}\right), & 1<\alpha<3/2\\
		\Theta \left(\max\{\frac{1}{n},{\frac{\log M}{\sqrt{Mna(n)}}}\}\right), & \alpha=1\\
		\Theta \left(\max\{\frac{1}{n},\frac{1}{\sqrt{Mna(n)}}\}\right), & \alpha<1\\
		\end{cases}
		\end{equation}
	\end{theorem}
	\begin{proof}  
		We prove that the average delay is given by (\ref{delayadhoc}).  The average throughput given in (\ref{throughputadhoc}) can be calculated easily by equation (\ref{tradeoff}). Substituting for the $p_j$'s in equation (\ref{D}) using the Zipf distribution, we obtain
		\begin{IEEEeqnarray}{+rCl+x*}
			\label{Dz}
			D&=&\frac{H_{\alpha}(m_1-1)}{H_{\alpha}(M)}+\frac{[H_{\frac{2\alpha}{3}}(m_2-1)-H_{\frac{2\alpha}{3}}(m_1-1)]^{3/2}}{\sqrt{nK'a(n)}H_{\alpha}(M)}\nonumber\\
			&+&
			\frac{H_{\alpha}(M)-H_{\alpha}(m_2-1)}{\sqrt{a(n)}H_{\alpha}(M)}.
		\end{IEEEeqnarray}
		where $K'=K-\frac{(m_1-1)}{2\log n}-\frac{(M-m_2+1)}{n} = \Theta(1)$.  Let the three expressions on the RHS of~\eqref{Dz} be denoted by $D_1$, $D_2$, and $D_3$, respectively.
		
		Clearly, $D_1= \Theta(1), \forall\alpha>0$. Also, if $a(n) =1$ then $m_2 =m_1+1$, and $D_2 = 0$. It can easily be shown that $D=\Theta(1)$, and $\lambda = W/n$, which coincides with the result of time-division with direct communication policy. Hence, we assume here that $a(n)<1$. By Lemma \ref{m1m2lemma}, we know that for $\alpha \leq 3/2$, $m_2=M+1$. Therefore, $D_3$ is zero, and $D=\Theta(\max\{1,D_2\})$. 
		
		For $\alpha<1$: 
		\begin{equation}
			D_2=\Theta\left(\frac{(m_2-1)^{3/2-\alpha}}{\sqrt{na(n)}M^{1-\alpha}}\right)=\Theta\left(\sqrt{\frac{M}{na(n)}}\right).
		\end{equation}
		
		For $\alpha=1$:
		\begin{equation}
			D_2=\Theta\left(\frac{(m_2-1)^{1/2}}{\log M\sqrt{na(n)}}\right)=\Theta\left({\frac{\sqrt{M}}{\log M\sqrt{na(n)}}}\right).
		\end{equation}
		
		For $1<\alpha<3/2$: similarly, we have
		\begin{equation}
			D_2=\Theta\left(\frac{M^{3/2-\alpha}}{\sqrt{na(n)}}\right).
		\end{equation}
		
		For $\alpha=3/2$: 
		\begin{equation}
			D_2=\Theta\left(\frac{(\log M)^{3/2}}{\sqrt{na(n)}}\right).
		\end{equation}
		%Here, we need to distinguish between the case where $m_2=\Theta(M)$ and the one in which $m_2=o(M)$. In the first case,
		For $\alpha>3/2$:  $D_2=\Theta(\frac{1}{\sqrt{na(n)}})=o(1)$. Also, as shown in the following, $D_3=o(1)$. Therefore, $D=\Theta(D_1)=\Theta(1)$. Now, if $m_2=M+1$ then $D_3=0$. Otherwise, $m_2\simeq\frac{2\alpha-3}{2\alpha}K{n}(a(n))^{1-\frac{3}{2\alpha}}$. Using straightforward calculation, it follows that
		\begin{equation}
			D_3=\Theta({\frac{m_2^{1-\alpha}}{\sqrt{a(n)}}}) = o(1).
		\end{equation}
	\end{proof}
	
	To get more intuition about these results, we can substitute $a(n) = 2\log n/n$, in (\ref{delayadhoc}) and (\ref{throughputadhoc}). We have
	\begin{equation}
	\label{delayadhoc2}
	D^*(n)=
	\begin{cases}
	\Theta(1), & \alpha>3/2\\
	\Theta(\log M), & \alpha=3/2\\
	\Theta (\frac{M^{3/2-\alpha}}{\sqrt{\log M}}), & 1<\alpha<3/2\\
	\Theta (\frac{\sqrt{M}}{(\log M)^{3/2}}), & \alpha=1\\
	\Theta (\sqrt{\frac{M}{\log M}}), & \alpha<1\\
	\end{cases}
	\end{equation}
	\begin{equation}
	\label{throughputadhoc2}
	\lambda^*(n)=
	\begin{cases}
	\Theta(\frac{1}{\log M}), & \alpha>3/2\\
	\Theta(\frac{1}{(\log M)^2}), & \alpha=3/2\\
	\Theta (\frac{M^{\alpha-3/2}}{\sqrt{\log M}}), & 1<\alpha<3/2\\
	\Theta (\sqrt{\frac{\log M}{M}}), & \alpha=1\\
	\Theta (\frac{1}{\sqrt{M\log M}}), & \alpha<1\\
	\end{cases}
	\end{equation}
	
%\begin{figure*}
%				\begin{subfigure}[t]{0.5\textwidth}
%		\centering
%		\includegraphics[width=.65\textwidth]{delayadhoc.pdf}
%		\subcaption{Scaling of network delay for various values of $\alpha$ vs. number of nodes,\\ for $\beta =0.9$.}
%		\label{delayadhocfig}
%	\end{subfigure}
%		\begin{subfigure}[t]{0.5\textwidth}
%		\centering
%		\includegraphics[width=.65\textwidth]{throughputadhoc.pdf}
%		\subcaption{The logarithmic scaling of network throughput for various values of $\alpha$ vs. number of nodes, for $\beta =0.9$.}
%		\label{throughputadhocfig}
%	\end{subfigure}
%\end{figure*}
%	
%	In Figure \ref{delayadhocfig}, we have plotted the theoretical results given in (\ref{delayadhoc2}), to demonstrate the scaling of the network delay under the achievable scheme for various values of $\alpha$.  Similarly, in Figure \ref{throughputadhocfig}, the scaling of per-node throughput of the achievable scheme is depicted, following (\ref{throughputadhoc2}), for different values of $\alpha$. For both figures, $\beta=0.9.$ Furthermore, the constants are normalized to focus on the scaling of the curves.
%	It is clear from the figures that the throughput and delay advantages offered by caching are both enhanced as $\alpha$ increases.  This makes intuitive sense: for larger $\alpha$, the content popularity distribution is dominated by the most popular content objects, which when cached, can lead to increased throughput and decreased delay at the same time.
%	
%	
%	
	\section{Heterogeneous Wireless Networks}
	\label{hybrid}
	Thus far, we have considered a pure ad hoc wireless network with caching, in which there are no base stations.  We now consider a more general heterogeneous wireless network environment with caching and show that the proposed model for ad hoc networks can be naturally extended to the heterogeneous case.   Consider a heterogeneous wireless network where, in addition to uniformly distributed wireless nodes, there are a number of base stations which are also uniformly distributed at random in the network area.   This models the scenario where smaller cells, e.g. femtocells, are deployed with random placement of base stations inside the network area~\cite{heter}.  The base stations are distinguished from the wireless nodes in that they are assumed to connect to the wired backbone, and thus are assumed to have access to all $M$ content objects.   Let $f(n)$ be the number of base stations, where $f(n)$ is a non-decreasing function of $n$. For our analysis, we assume $f(n)=\Theta(n^{\mu})$, where $0\leq\mu<1$.
	
	We assume that each wireless node is assigned to the closest base station in Euclidean distance. Thus, the network area is divided into $f(n)$ cellular regions.  If the size of each cellular region is large compared to the transmission range $r(n)$ (equivalently $a(n)$) of the wireless nodes, then a wireless node transmits to its assigned base station via multi-hop relaying through other wireless nodes. 
	
	We now consider a transmission and caching scheme for the heterogeneous wireless network, which is similar to the scheme considered for the ad hoc case.  That is, the network area is divided into $a^{-1}(n)$ squared cells each with area $a(n)$.  Based on a TDM scheme, each node, including base stations, transmit packets over the shared channel, subject to the Protocol Model.  For simplicity, we assume all the nodes, including base stations, have the same transmission range, $r(n)$. Note that this is a reasonable assumption when considering femtocells.
	
	Each wireless node can request contents from its assigned base station through multi-hop relaying.  Each wireless node requests content $m$ with probability $p_m$.  If the closest wireless holder of content $m$ is closer to the requesting node than the node's assigned base station, then the content is retrieved from the closest wireless holder. Otherwise, it is retrieved from the base station.
	%\begin{figure}
	%  \centering
	%  \includegraphics[width=.5\textwidth]{networkfig.jpg}
	%  \caption{Random placement of the base stations and the corresponding cellular regions: each diamond represents a base stations. Region boundaries correspond to contours of maximum received power.}
	%  \label{networkfig}
	%\end{figure}
	
	Similar to the previous sections, we assume that the $X_m$ wireless holders of content $m$ are uniformly distributed in the network area. Since we are interested in evaluating the performance of the wireless network, we assume that all requests for content, upon reception at base stations, are satisfied immediately (i.e. a Data Packet is generated immediately).  In other words, we do not consider the delay within the wired backbone network. 
	
	Unlike the pure ad hoc case in which we need to have at least one copy of each content object in the caches of the wireless nodes to satisfy all the requests, for the proposed heterogeneous network we relax this restriction due to the presence of the base stations. As a result, the number of content types can exceed the number of nodes. i.e., $\beta$ can be $\geq 1$.

	As in Lemma \ref{lemma2}, we can show that the average length of the $L_{H,R}(i,m)$ line connecting the requesting node $i$ to the closest cache of content $m$ (either a wireless holder or a base station) is given by:
	\begin{equation}
	E\left[|L_{H,R}(i,m)|\right]=\Theta\left(\frac{1}{\sqrt{X_m+f(n)}}\right).\quad 
	\end{equation}
	
	Consequently, the average of number of hops along the $L_{H,R}$ line is w.h.p. 
	\begin{equation}
	\label{formhopsH}
	E[H_{i,m_i}]
	= \Theta \left(  \max \left\{1,\frac{1}{\sqrt{a(n)(X_m+f(n))}} \right\} \right).
	\end{equation}
	
	Using an approach similar to that in the proof of Lemma \ref{lemma3}, we see that for \smash{$a(n)\geq 2{\log n/n}$}, the number of $L_{H,R}$ lines passing through each cell (of area $a(n)$) is 
	$$\Theta \left( n\sum_{m=1}^{M} p_m \max \left\{ a(n), \sqrt{\frac{a(n)}{X_m+f(n)}}   \right\}\right)  w.h.p.$$ 
	
	Therefore, the throughput and the delay of the achievable scheme for the heterogeneous network model are given by:
	\begin{equation}
	\label{throughputbound_h}
	\lambda(n)= \Theta \left(\frac{1}{n\sum_{m=1}^M  p_m \max{\{a(n),\sqrt{\frac{a(n)}{X_m+f(n)}} \}}} \right) w.h.p. 
	\end{equation}
	
	\begin{equation}
	\label{delayH}
	D=\Theta\left(\sum_{m=1}^{M} p_m \max{\left\lbrace 1,\frac{1}{\sqrt{a(n)(X_m+f(n))}}\right\rbrace}\right)w.h.p. 
	\end{equation}
	Combining the equations (\ref{throughputbound_h}) and (\ref{delayH}), we obtain the same throughput and delay relation as in the ad hoc case given in (\ref{tradeoff}).
	
	Next, we optimize the throughput and delay of the achievable scheme for the heterogeneous network scenario by choosing the appropriate $(X_m)_{m=1}^M$ . Note that here the constraints on $X_m$ are $0\leq X_m \leq a^{-1}(n)-f(n)$, as larger $X_m$'s do not change the order of the throughput or delay.  Thus, the optimization problem is
	\begin{equation}
	\begin{cases}
	\ \min_{\{X_m\}} \sum_{m=1}^M \frac{p_m}{\sqrt{a(n)(X_m+f(n))}} \\
	\ \text{subject to:}\\
	\ \sum_{m=1}^M X_m \leq nK \\
	\ 0 \leq X_m \leq a^{-1}(n)-f(n)  \qquad \text{for } m=1,2,\ldots,M\\
	\end{cases}
	\label{optimizationH}
	\end{equation}
	
	Since the objective function is strictly convex, we are assured that there is a unique global minimum. Defining the non-negative Lagrange multipliers $\lambda$ for the constraint $\sum_{m=1}^{M}X_m\leq nK$, and taking into account the  constraint $0 \leq X_m \leq a^{-1}(n)-f(n)$, the necessary conditions for a minimum of $D$ with respect to $X_m$, $\forall m \in M$ are given
	\begin{equation}
	\frac{\partial D}{\partial X_m}
	\begin{cases}
	\ \leq -\lambda \qquad \text{if } X_m=a^{-1}(n)-f(n) \\
	\ = -\lambda \qquad \text{if } 0<X_m<a^{-1}(n)-f(n)\\
	\ \geq -\lambda \qquad \text{if } X_m=0 \\
	
	\end{cases}
	\label{necessaryH1}
	\end{equation}
	
	Given the Zipf distribution, let $\mathcal{M}_1=\{1,2,\ldots,m_1-1\}$ be the set of content objects such that $X_m=a^{-1}(n)-f(n)$ for $m\in\mathcal{M}_1$. Similarly, let $\mathcal{M}_2=\{m_1,m_1+1,\ldots,m_2-1\}$ and $\mathcal{M}_3=\{m_2,m_2+1,\ldots,M\}$ be the set of contents such that $0<X_m<a^{-1}(n)-f(n)$ for $m\in\mathcal{M}_2$, and $X_m=0$ for $m\in\mathcal{M}_3$, respectively. From (\ref{necessaryH1}), we have $\forall m \in M$
	\begin{equation}
	\frac{p_m}{2\sqrt{a(n)(X_m+f(n))^3}}
	\begin{cases}
	\ \geq \lambda \qquad  \forall m\in\mathcal{M}_1\\
	\ = \lambda \qquad \forall m\in\mathcal{M}_2\\
	\ \leq \lambda \qquad \forall m\in\mathcal{M}_3 \\
	
	\end{cases}
	\label{necessaryH2}
	\end{equation}

	Using the equality for the case $\forall m\in\mathcal{M}_2$, we obtain
	\begin{equation}
	\label{m1/m2h}
	\frac{m_1}{m_2}\simeq (a(n)f(n))^{\frac{3}{2\alpha}}.
	\end{equation}
	
	From (\ref{necessaryH2}), we have $\lambda>0$ and hence, $\sum_{m=1}^{M}X_m=nK$. Combining this with (\ref{m1/m2h}), we can derive $m_1$ and $m_2$. The optimal number of holders of content $m$, $X_m^*$, is then given by
	
	\begin{equation}
	\label{Xmh}
	X_{m}^*=
	\begin{cases}
	\ a^{-1}(n)-f(n),\qquad\qquad  m=1,2,\ldots, m_1-1\\
	\ \frac{p_m^{2/3}}{\sum_{j=m_1}^{m_2-1}p_j^{2/3}} nK'-f(n), m=m_1,\ldots,m_2-1\\
	\ \quad 0, \qquad\qquad\qquad\qquad m=m_2,\ldots,M
	\end{cases}
	\end{equation}
	where $K'\triangleq K-(m_1-1)\frac{a^{-1}(n)}{n}+(m_2-1)\frac{f(n)}{n}$. Hence, the average delay is w.h.p.
	\begin{equation}
	D^*(n)=\Theta \left(\sum_{j=1}^{m_1-1}p_j+\frac{\left(\sum_{j=m_1}^{m_2-1}p_j^{2/3}\right)^{3/2}}{\sqrt{a(n)nK'}}+\frac{\sum_{j=m_2}^{M}p_j}{\sqrt{f(n)a(n)}}\right).
	\label{Dh}
	\end{equation}
	
	We can now apply techniques similar to the one used in the ad hoc case in order to estimate the indices $m_1$ and $m_2$, and then compute the scalings of the delay and throughput. So far we have considered $a(n) \geq 2 \log n/n$ to be a general parameter resulting in a trade-off between the throughput and delay of the network: as $a(n)$ increases (decreases), both throughput and delay of the network decrease (increase). In this section, we consider a single point of this trade-off where $a(n) = 2 \log n/n$, as this will give us more intuitive formulas for delay and throughput. The generalization of this result is a straightforward calculation following the approach of the ad hoc case.  Following this, we can estimate the indices $m_1$ and $m_2$ as follows.
	\begin{lemma}
		\label{m1m2h}
		Taking $n\rightarrow\infty$, $m_1$ and $m_2$ scales as:
		\begin{equation}
		\label{m1h}
		m_1=
		\begin{cases}
		\ \Theta(\log n) & \alpha>3/2\\
		\ \Theta(1) & \alpha=3/2\\
		\ \text{converging to }1 & \alpha<3/2
		\end{cases}
		\end{equation}
		\begin{equation}
		\label{m2h}
		m_2=
		\begin{cases}
		\ \min\{M+1,\Theta\left((\frac{n}{f(n)})^{\frac{3}{2\alpha}}(\log n)^{1-\frac{3}{2\alpha}}\right)\} & \alpha>3/2\\
		\ \min\{M+1,\Theta\left(\frac{n}{f(n)\log n}\right)\} & \alpha=3/2\\
		\ \min\{M+1,\Theta\left(\frac{n}{f(n)}\right)\} & \alpha<3/2
		\end{cases}
		\end{equation}
	\end{lemma}
	\begin{proof}
		Refer to Appendix \ref{app_m1m2h}.
		\end{proof}
	
	We now compute the throughput and delay of the proposed heterogeneous network model as follows. Note that part \ref{them:1} of Theorem \ref{hybridtheorem}, considers the case where $m_2=M+1$. For $\alpha\leq3/2$ this happens when $\beta< 1-\mu$, or equivalently $f(n)=o(\frac{n}{M})$ and $f(n)\geq1$. For $\alpha>3/2$, $m_2=M+1$ if $\beta\leq\frac{3}{2\alpha}(1-\mu)$, or equivalently $f(n)=O(\frac{n}{M^{2\alpha/3}})$ and $f(n)\geq1$. 
	On the other hand, part \ref{them:2} of Theorem \ref{hybridtheorem} shows the performance of the network when $m_2\leq M$. For $\alpha\leq3/2$ this happens when $\beta\geq 1-\mu$, or equivalently $f(n)=\Omega(\frac{n}{M})$ and $f(n)\geq1$. In addition, for $\alpha>3/2$, $m_2\leq M$ if $\beta>\frac{3}{2\alpha}(1-\mu)$, or equivalently $f(n)=\omega(\frac{n}{M^{2\alpha/3}})$ and $f(n)\geq1$. Note that for any value of $\alpha$, if $f(n)=\Omega(\frac{n}{M})$ and $f(n)\geq1$ (or equivalently $\mu\geq \max\{0,1-\beta\}$), then the heterogeneous network performance follows (\ref{hybriddelay}) and (\ref{throughputhybrid}).
	\begin{theorem}
		\label{hybridtheorem}
		For $a(n) = 2 \log n/n$,
		\begin{enumerate}
			\item \label{them:1}The throughput and delay performance of the achievable scheme for the heterogeneous network, when $m_2=M+1$ and the content popularity distribution follows the Zipf distribution, is the same as given in (\ref{throughputadhoc}) and (\ref{delayadhoc}), respectively.
			\item  \label{them:2}
			The throughput and delay of the achievable scheme, when $m_2\leq M$, are w.h.p.:
			\begin{equation}
			\label{hybriddelay}
			D^*(n)=
			\begin{cases}
			\Theta(1) & \alpha>3/2\\
			\Theta\left(\log n\right) & \alpha=3/2\\
			\Theta \left(\frac{ (\frac{n}{f(n)})^{3/2-\alpha}}{\sqrt{\log n}}\right) & 1<\alpha<3/2\\
			\Theta \left(\sqrt{\frac{n}{f(n)\log n}}\right) & \alpha\leq 1
			
			\end{cases}
			\end{equation}
			\begin{equation}
			\label{throughputhybrid}
			\lambda^*(n)=
			\begin{cases}
			\Theta(\frac{1}{\log n}) & \alpha>3/2\\
			\Theta\left(\frac{1}{(\log n)^2}\right) & \alpha=3/2\\
			\Theta \left(\frac{1}{\sqrt{\log n} (\frac{n}{f(n)})^{3/2-\alpha}}\right) & 1<\alpha<3/2\\
			\Theta \left(\sqrt{\frac{f(n)}{n\log n}}\right) & \alpha\leq 1\\
			
			\end{cases}
			\end{equation}
		\end{enumerate}
	\end{theorem} 
	\begin{proof}
We compute the average delay. The average throughput follows by (\ref{tradeoff}). Substituting for the $p_j$'s in equation (\ref{Dh}) using the Zipf distribution, we have
\begin{equation}
	\label{Dzh}
	\begin{split}
		D=\frac{H_{\alpha}(m_1)}{H_{\alpha}(M)}+\frac{[H_{\frac{2\alpha}{3}}(m_2-1)-H_{\frac{2\alpha}{3}}(m_1-1)]^{3/2}}{\sqrt{K'\log n}H_{\alpha}(M)}\\
		+
		\sqrt{\frac{n}{f(n)\log n}}\cdot\frac{H_{\alpha}(M)-H_{\alpha}(m_2-1)}{H_{\alpha}(M)}.
	\end{split}
\end{equation}
where $K'\rightarrow K-\frac{(m_1-1)}{\log n}$ as $n\rightarrow \infty$. Similar to the proof of Theorem \ref{theorem5}, let the three expressions on the RHS of~\eqref{Dzh} be denoted by $D_1$, $D_2$, and $D_3$, respectively.
Moreover, when $m_2=M+1$, $D_3=0$. Hence, the equation (\ref{Dzh}) is simplified to equation (\ref{Dz}), given that $m_2=M+1$. As shown in (\ref{m2}), this always holds for $\alpha\leq 3/2$. In addition, for $\alpha>3/2$, if we assign $m_2=M+1$, we still get the same result as shown in the proof of Theorem \ref{theorem5}. 

Now we prove the results for the second part of the theorem, where $m_2\leq M$.
By Lemma \ref{m1m2h}, we know for $\alpha\leq 3/2$, $K'\rightarrow K$ and $D_1=o(1)$. For $\alpha>3/2$, $D_1=\Theta(1)$. 

For $\alpha<1$: by (\ref{m2h}), $m_2=\Theta(\frac{n}{f(n)})$. Following (\ref{Dzh}),
\begin{equation}
	D_2\simeq \frac{n^{(1-\mu)(3/2-\alpha)}}{\sqrt{\log n}M^{1-\alpha}}.
\end{equation}
\begin{equation}
	\label{D3a1}
	D_3\simeq\frac{n^{\frac{1-\mu}{2}}}{\sqrt{\log n}}.
\end{equation}

It can easily be shown that $D_2=o(D_3)$. Thus, $D=\Theta(D_3)$.

For $\alpha=1$: Similarly, by using (\ref{Dzh}), it follows that $D_3$ is given by (\ref{D3a1}). For $D_2$ we have
\begin{equation}
	D_2\simeq \frac{n^{(1-\mu)(3/2-\alpha)}}{\sqrt{\log n}\log M}.
\end{equation}
Now since $\log M=\Theta(\log n)$, we have $D_2=\Theta\left(\frac{n^{\frac{1-\mu}{2}}}{(\log n)^{3/2}}\right)$. Clearly, $D_2=o(D_3)$. Hence, $D=\Theta(D_3)$.

For $1<\alpha<3/2$: By using the same technique as in the previous part, we can see that $D_3=\Theta(D_2)$ and therefore, $D=\Theta(D_2)$. we have
\begin{IEEEeqnarray}{+rCl+x*}
	D_2&\simeq& \frac{m_2^{3/2-\alpha}}{\sqrt{\log n}} 
	= \Theta\left(\frac{n^{(1-\mu)(3/2-\alpha)}}{\sqrt{\log n}}\right).
\end{IEEEeqnarray} 
\begin{IEEEeqnarray}{+rCl+x*}
	D_3&\simeq& \frac{m_2^{1-\alpha}\cdot n^{\frac{1-\mu}{2}}}{\sqrt{\log n}}   
	= \Theta\left(\frac{n^{(1-\mu)(3/2-\alpha)}}{\sqrt{\log n}}\right).
\end{IEEEeqnarray} 

For $\alpha=3/2$: using Lemma \ref{m1m2h}, it follows from (\ref{Dzh}) that
\begin{equation}
	D_2\simeq\frac{\left(\log \frac{n^{1-\mu}}{\log n}\right)^{3/2}}{\sqrt{\log n}}=\Theta(\log n).
\end{equation}
\begin{equation}
	D_3\simeq \frac{m_2^{-1/2}\cdot n^{\frac{1-\mu}{2}}}{\sqrt{\log n}}
	=\Theta(1).
\end{equation}   
Therefore, $D=\Theta(D_2)$.

For $\alpha>3/2$:  using a similar calculation, we have
\begin{equation*}
	D_2\simeq\frac{m_1^{3/2-\alpha}}{\sqrt{\log n}}=o(1).
\end{equation*}
and,
\begin{equation}
	\label{D3a32}
	D_3\simeq\frac{m_2^{1-\alpha}n^{\frac{1-\mu}{2}}}{\sqrt{\log n}}=o(1).
\end{equation}
To show the last equation in (\ref{D3a32}), let's consider the power of $n$ in $D_3$: $\frac{3}{2\alpha}(1-\mu)(1-\alpha)+\frac{1-\mu}{2}=(1-\mu)(\frac{3}{2\alpha}-1)<0$. Hence, 
$D_3\rightarrow 0$ as $n\rightarrow \infty$. Therefore, $D=\Theta(D_1)=\Theta(1)$.

	\end{proof}
	Comparing the results for the heterogeneous network in Theorem \ref{hybridtheorem} with those for the pure ad hoc network given in Theorem \ref{theorem5}, for $ \alpha \neq 1$ and  $a(n) = \Theta(\frac{\log n}{n})$, we conclude that 
	the number of base stations in the network needs to be greater than $\frac{n}{M}=n^{1-\beta}$ to improve the order of the performance metrics (throughput and delay). For the scenario where $\beta\geq 1$, this condition reduces to $f(n)\geq 1$. In other words, if $\beta\geq 1$, the heterogeneous network always outperforms the pure ad hoc network. Also, note that for $\alpha\geq 3/2$, the performance of the heterogeneous network is the same as that for the pure ad hoc case. Intuitively, this is because for large $\alpha$'s, the majority of content requests are for the most popular content objects, hence, caching the most popular content objects will almost eliminate the need for base stations.
	\begin{figure*}

		\begin{subfigure}[t]{0.5\textwidth}
				 	 		  	  	\centering
			\includegraphics[width=.75\textwidth]{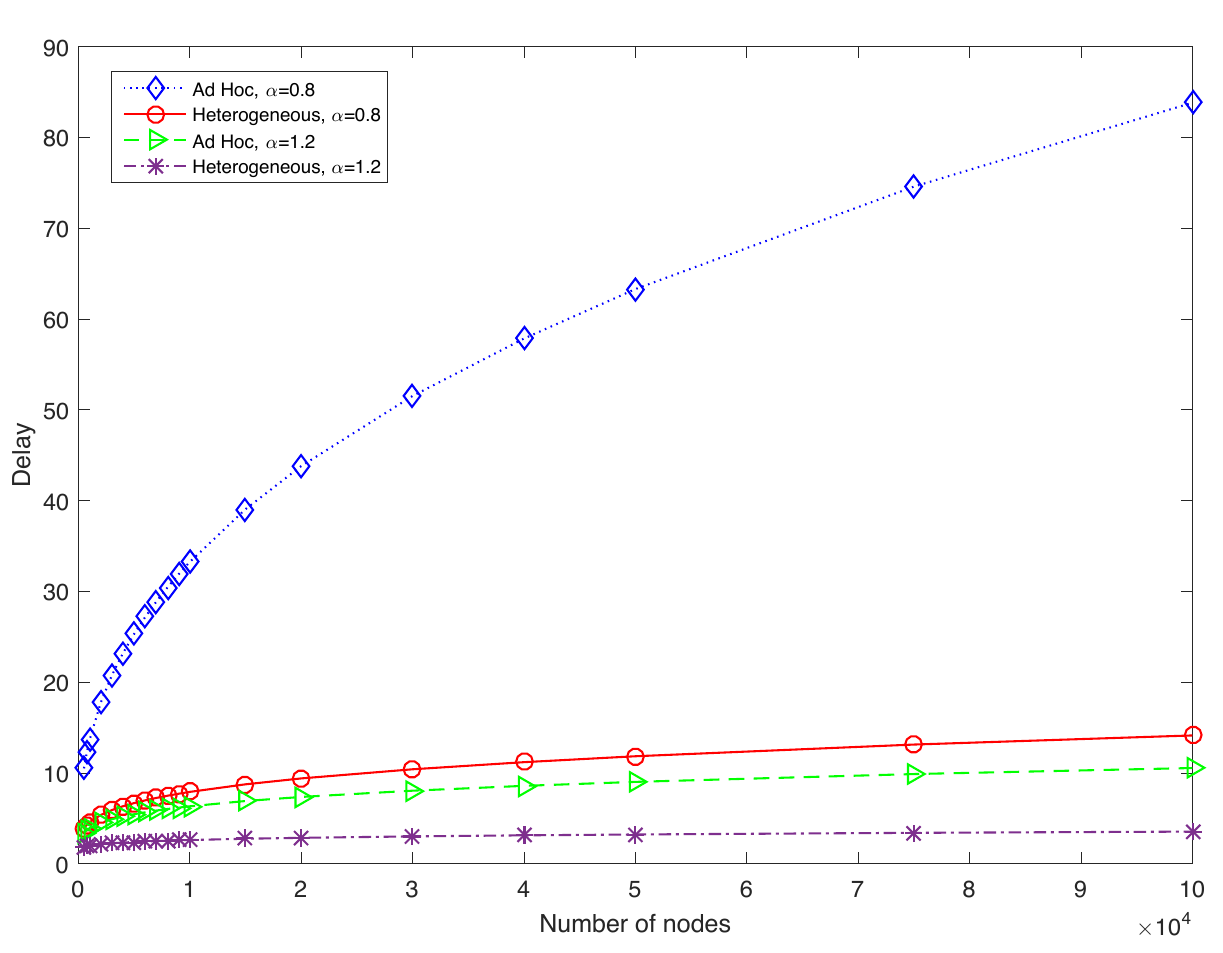}
			\subcaption{The scaling of delay for the heterogeneous and ad hoc network models\\ for various values of $\alpha$, vs. number of nodes, for $\beta =0.9$, and for $\mu =0.4$.}
			\label{delayheterfig}
		\end{subfigure}
		\begin{subfigure}[t]{0.5\textwidth}
				 	 		  	  	\centering
			\includegraphics[width=.75\textwidth]{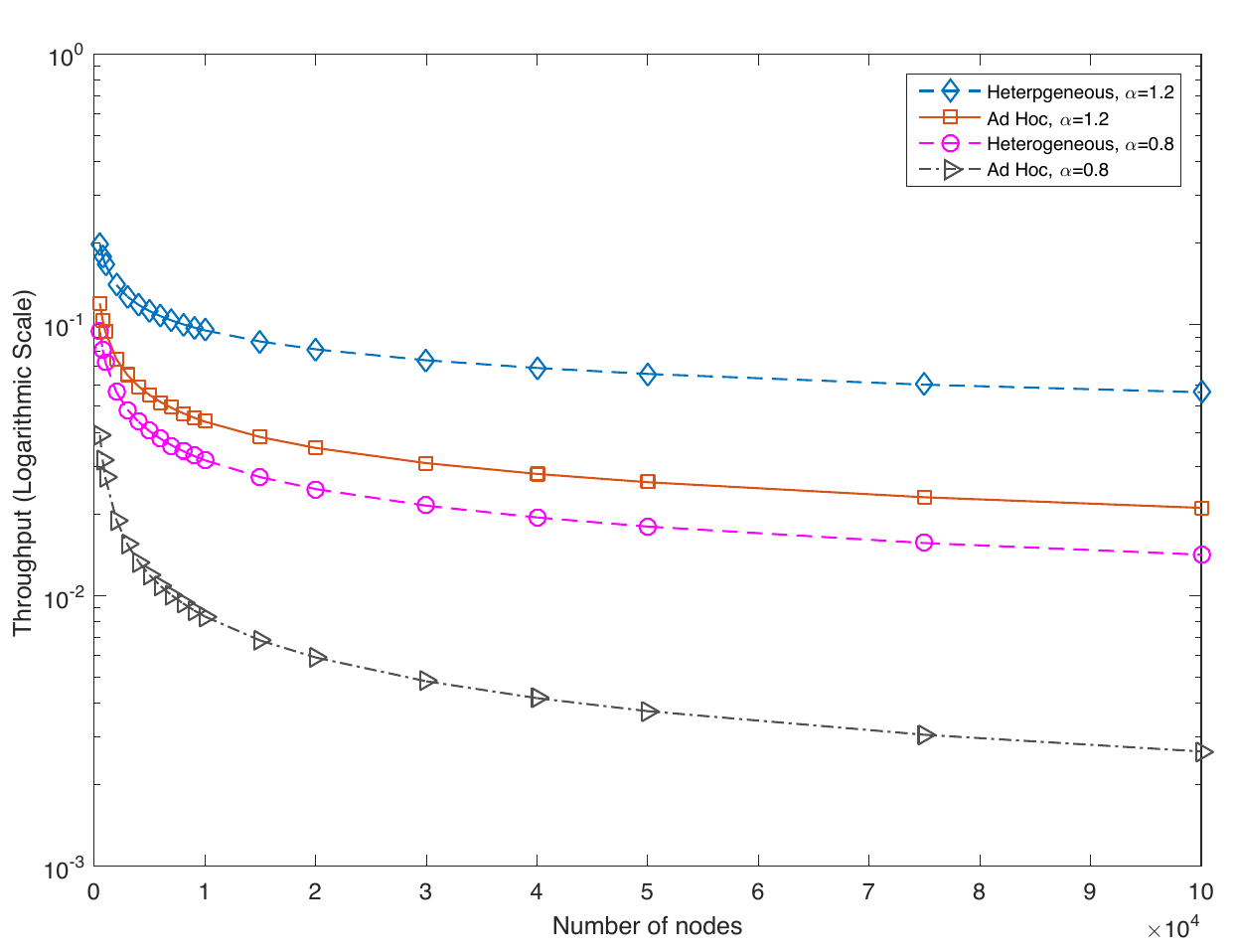}
			\subcaption{The logarithmic scaling of per-node throughput for the heterogeneous and ad hoc network models for various values of $\alpha$, vs. number of nodes, for $\beta =0.9$, and, for $\mu =0.4$.}
			\label{throughputheterfig}
		\end{subfigure}
			\begin{subfigure}[t]{0.5\textwidth}
					 	 		  	  	\centering
				\includegraphics[width=.75\textwidth]{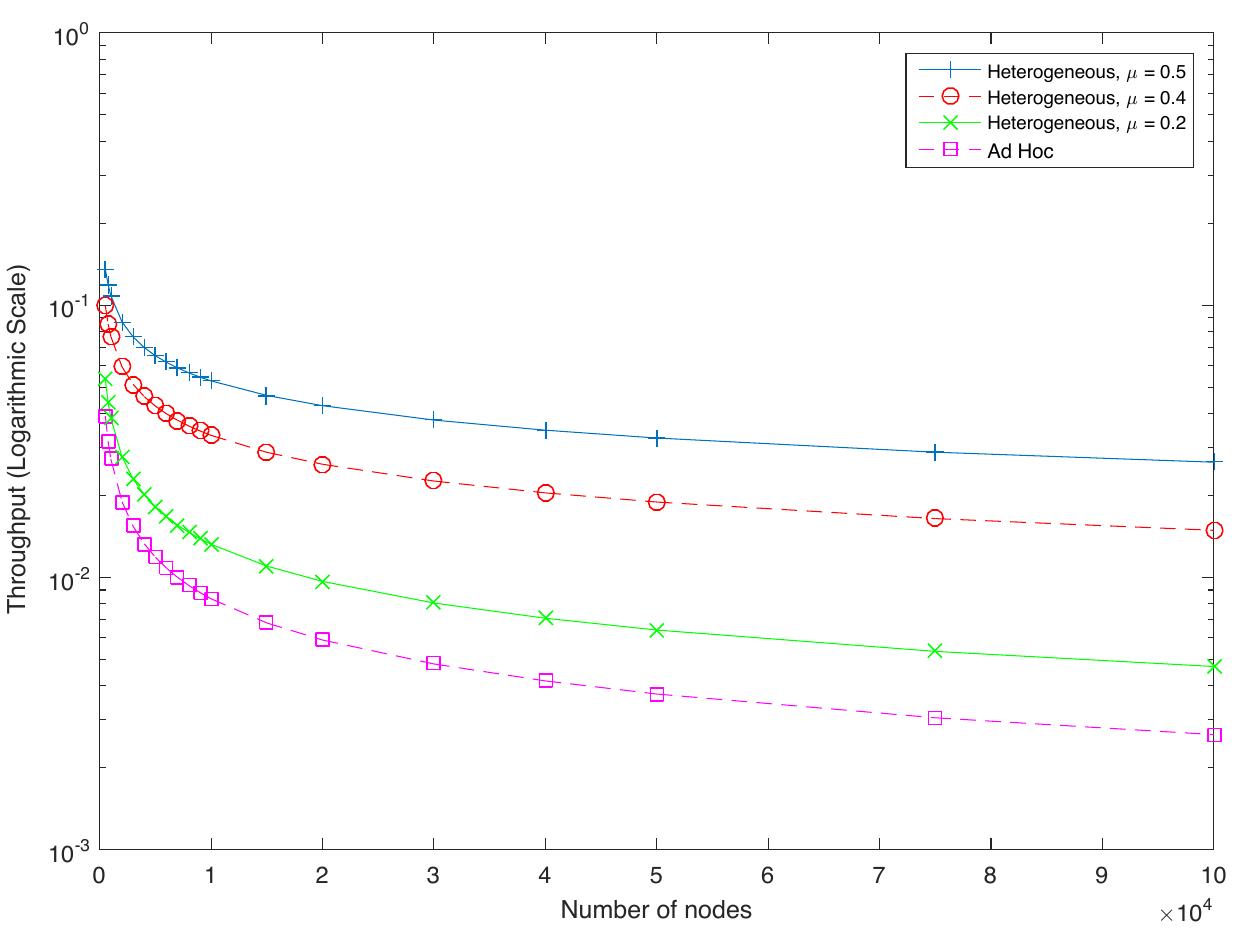}
				\subcaption{The logarithmic scaling of per-node throughput for the heterogeneous\\ and ad hoc network models for various values of $\mu$, vs. number of \\nodes, for $\beta =0.9$, and, for $\alpha =0.8$.}
				\label{throughputmuheterfig}
			\end{subfigure}
			\begin{subfigure}[t]{0.5\textwidth}
					 	 		  	  	\centering
			\includegraphics[width=.75\textwidth]{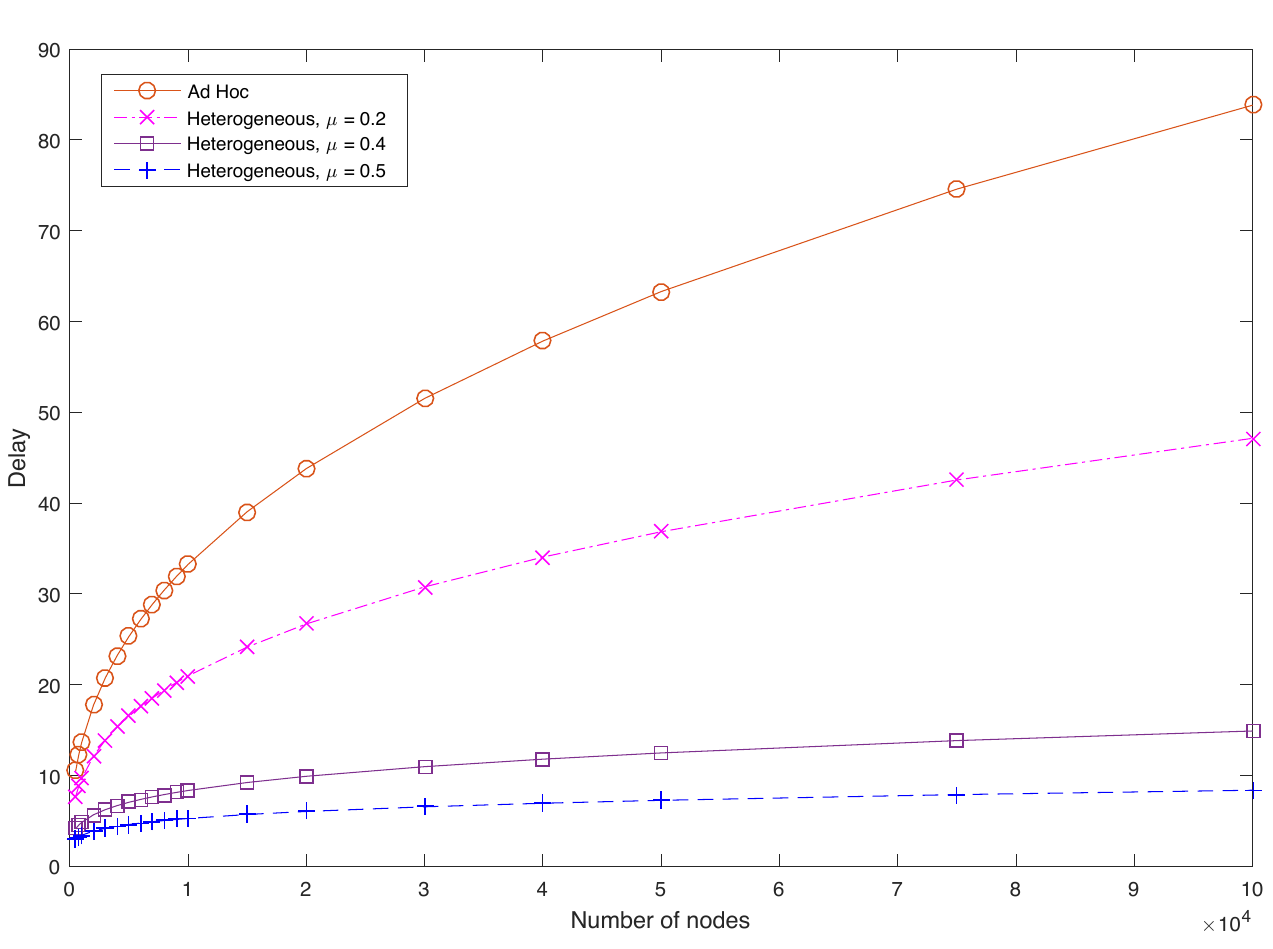}
			\subcaption{The scaling of delay for the heterogeneous and ad hoc network models for various values of $\mu$, vs. number of nodes, for $\beta =0.9$, and, for $\alpha =0.8$.}
			\label{delaymuheterfig}
		\end{subfigure}
	\end{figure*}

We have plotted the theoretical results given in (\ref{hybriddelay}) and (\ref{throughputhybrid}) in Figures \ref{delayheterfig} and \ref{throughputheterfig} , respectively, to demonstrate the scaling of the network delay and per-node throughput for $\alpha=0.8$ and $\alpha=1.2$. The constants are normalized to focus on the scaling of the curves.
In addition, we have plotted the performance of the ad hoc network model for the same values of $\alpha$. In both figures, $\beta=0.9$ and $f(n)=n^{0.4}$.  Note that for $\alpha\geq3/2$, the performance of the heterogeneous network is the same in order as that for the ad hoc case. In Figures \ref{throughputmuheterfig} and \ref{delaymuheterfig}, the scaling of the per-node throughput and network delay is shown for $\alpha = 0.8$, $\beta =0.9$, and various values of $\mu$, along with the corresponding scaling for the pure ad hoc case. As predicted, by adding more base stations to the network, the performance of the network, both in terms of throughput and delay, is improved. 

\section{Conclusions}
We have investigated the asymptotic behavior of wireless caching networks.   We presented an achievable caching and transmission scheme whereby requesters retrieve content from the holder which is closest in Euclidean distance.  We established the throughput and delay scaling of the achievable caching/transmission scheme, and showed that the throughput and delay performance are order-optimal within a class of schemes.   We then optimized the caching strategy to simultaneously minimize the average network delay and maximize the network throughput.  Using the optimal caching strategy, we evaluated the network performance under a Zipf content popularity distribution. 

Furthermore, we investigated heterogeneous wireless networks where, in addition to wireless nodes, there are a number of base stations uniformly distributed at random in the network area.  We showed that in order to achieve a better performance in a heterogeneous network in the order sense, the number of base stations needs to be greater than the ratio of the number of nodes to the number of content types. For the case where the number of content objects is greater than the number of wireless nodes, this condition reduces to having at least one base station in the network. In addition, we demonstrated that for the Zipf content popularity distribution with exponent $\alpha\geq 3/2$, the performance of the wireless ad hoc network is of the same order as for the heterogeneous wireless network, independent of number of base stations.
\appendix
\subsection{Proof of Lemma \ref{lemma2}}
\label{a}

Since the holders are independently and uniformly distributed, the probability that no holder is within distance less than or equal to $\tau$ of the requester is $\Pr(d \geq \tau)=(1-\pi\tau^2)^{X_m}$ for $0\leq \tau \leq 1/\sqrt{\pi}$. Therefore, the average distance from the requester to the closest holder is

\begin{IEEEeqnarray*}{rCl}
	E[d] & = & \int_{0}^{\infty}\Pr(d \geq \tau) \mathrm{d}\tau \\
	& = & \int_{0}^{\frac{1}{\sqrt{\pi}}}(1-\pi\tau^2)^{X_m} \mathrm{d}\tau.
\end{IEEEeqnarray*}
Using a change of variable $\sqrt{\pi}\tau = \cos \theta$ and applying integration by parts, we have
\begin{IEEEeqnarray}{+rCl+x*}
	E[d] & = & \frac{1}{\sqrt{\pi}}\int_{0}^{\frac{\pi}{2}}(\sin \theta )^{2X_m+1}  \mathrm{d}\theta \nonumber\\
	& = & \frac{1}{\sqrt{\pi}}\frac{2X_m}{2X_m+1}\cdot \frac{2X_m-2}{2X_m-1} \cdot \ldots \frac{2}{3} \cdot \int_{0}^{\frac{\pi}{2}} \sin \theta \mathrm{d} \theta \label{formula6}\IEEEeqnarraynumspace \text{    } \\
	& = &\frac{1}{\sqrt{\pi}}\frac{2X_m}{2X_m+1}\cdot \frac{2X_m-2}{2X_m-1} \cdot \ldots \frac{2}{3}\\
	& = & \Theta(\frac{1}{\sqrt{X_m}}). \label{result}
\end{IEEEeqnarray}
where (\ref{formula6}) is derived from
\begin{equation}
\int \sin^n x \mathrm{d} x= - \frac{1}{n} \sin^{n-1} x \cos x + \frac{n-1}{n} \int \sin^{n-2} x \mathrm{d} x.
\end{equation}
(\ref{result}) is followed from the fact that 
\begin{equation}
\frac{n_2}{n_1+1} \leq \left(\frac{g(n_1)}{g(n_2)}\right)^2 \leq \frac{n_2+1}{n_1},
\end{equation}
where 
\begin{equation}
g(n)=\frac{n-1}{n}\cdot \frac{n-3}{n-2} \cdot \ldots \frac{2}{3},
\end{equation}
and $n_1$ and $n_2$ are  two arbitrary odd integers. Therefore, $g(2X_m+1)=\Theta(1/\sqrt{X_m})$.

\subsection{Proof of Lemma \ref{lemma_hops}}
\label{app_hops}
{We compute the result for $E[H_{i,m}]$.  The same argument may be used to find $E[H'_{i,m}]$.
}
\color{black}
To compute $E[H_{i,m}]$, we consider the case where the holder is within one hop of the requester, and the case where the holder is farther than one hop away.  We have
\begin{IEEEeqnarray*}{+l+x*}
	E[H_{i,m}]=\\
	E[H_{i,m}||L_{H,R}(i,m)|\leq \sqrt{a(n)}]\Pr(|L_{H,R}(i,m)|\leq \sqrt{a(n)})\\
	+ E[H_{i,m}||L_{H,R}(i,m)|> \sqrt{a(n)}]\Pr(|L_{H,R}(i,m)|> \sqrt{a(n)}).
\end{IEEEeqnarray*}

Clearly, $E[H_{i,m}||L_{H,R}(i,m)|\leq \sqrt{a(n)}]=1$. Also, since the side-length of each cell is $\sqrt{a(n)}$, it can be shown that $E[H_{i,m}||L_{H,R}(i,m)|> \sqrt{a(n)}]=\Theta({E[|L_{H,R}(i,m)|]}/{\sqrt{a(n)}})=\Theta(1/\sqrt{a(n)X_{m}})$. 

Letting $\alpha(n) \equiv \Pr(|L_{H,R}(i,m)|>\sqrt{a(n)} )$, it follows that
\begin{equation}
\label{Ehops}
E[H_{i,m}]=\Theta\left(1+\left[\frac{1}{\sqrt{a(n)X_m}}-1\right]\alpha(n)\right).
\end{equation}
Note that $\alpha(n)=\Pr(d>\sqrt{a(n)})=(1-\pi a(n))^{X_m}$.  Expanding $\alpha(n)$ using the binomial form, and noting that $\left(\frac{n}{k}\right)^k\leq \binom{n}{k}\leq \frac{n^k}{k!}$, for $n\geq k \geq 1$, we have
\begin{equation}
\label{hopbounds}
1+\sum_{i=1}^{X_m}(-1)^i\frac{(\pi a(n)X_m)^i}{i^i}\leq \alpha(n)\leq e^{-\pi a(n)X_m}.
\end{equation}
Now, as $n\rightarrow \infty$, for $X_m=\omega(1/a(n))$, $e^{-\pi a(n)X_m}\rightarrow 0$, and hence $\alpha(n)\rightarrow 0$, implying that $E[H_{i,m}]=1$. For $X_m=\Theta(1/a(n))$, both bounds in (\ref{hopbounds}), and consequently $\alpha(n)$, are constant, leading to $E[H_{i,m}]=\Theta(1)$.  On the other hand, for $X_m=o(1/a(n))$, $a(n)X_m\rightarrow 0$, resulting in both bounds in (\ref{hopbounds}) converging to 1, as $n\rightarrow \infty$. Substituting $\alpha(n)=1$ in (\ref{Ehops}) gives $E[H_{i,m}]=\Theta(\frac{1}{\sqrt{a(n)X_m}})$. Therefore, the average number of hops can be re-written as 
\begin{equation}
E[H_{i,m}]= \Theta \left(\max{\left\lbrace\frac{1}{\sqrt{a(n)X_{m}}},1\right\rbrace}\right) \; w.h.p. 
\end{equation}

	\subsection{Proof of Lemma \ref{m1m2lemma}}
	\label{app_m1m2lemma}
	As $M=o(n)$, then $ M-m_2=o(n)$. Therefore, $K'\rightarrow K-{(m_1-1)}\frac{a^{-1}(n)}{n}$ as $n\rightarrow \infty$. Clearly, $K' = \Theta (1)$, hence, $m_1 = O(na(n))$. Now, by definition, $m_1$ is the smallest index for which the number of holders is less than $a^{-1}(n)$. That is, $X_{m_1}<a^{-1}(n)$. Using (\ref{Xm}), it follows that
	\begin{equation}
	\label{f1}
	nK'a(n) <m_1^{\frac{2\alpha}{3}}[H_{\frac{2\alpha}{3}}(m_2-1)-H_{\frac{2\alpha}{3}}(m_1-1)].
	\end{equation}
	Now, if $m_1>1$, attempting to decrease the index $m_1$ by one would result in $$\frac{p_{m_1-1}^{2/3}}{\sum_{j=m_1-1}^{m_2-1}p_j^{2/3}} nK'\geq a^{-1}(n).$$
	Hence, we have
	\begin{equation}
	\label{f2}
	nK'a(n)  \geq(m_1-1)^{\frac{2\alpha}{3}}[H_{\frac{2\alpha}{3}}(m_2-1)-H_{\frac{2\alpha}{3}}(m_1-2)].
	\end{equation}
	Hence, for $m_1>1$, an approximation of $m_1$ can be obtained from:
	\begin{equation}
	\label{se1}
	nK'a(n) \simeq m_1^{\frac{2\alpha}{3}}[H_{\frac{2\alpha}{3}}(m_2-1)-H_{\frac{2\alpha}{3}}(m_1-1)].
	\end{equation} 
	
	Similarly, by the definition of $m_2$, we know $X_{m_2-1}>1$ 
	\begin{equation}
	\label{f3}
	nK'>(m_2-1)^{\frac{2\alpha}{3}}[H_{\frac{2\alpha}{3}}(m_2-1)-H_{\frac{2\alpha}{3}}(m_1-1)].
	\end{equation}
	Now if $m_2 \leq M$, attempting to increase the index $m_2$ by one would lead to  $$\frac{p_{m_2}^{2/3}}{\sum_{j=m_1}^{m_2}p_j^{2/3}} nK' \leq 1.$$ 
	Thus, it follows that
	\begin{equation}
	\label{f4}
	n K' \leq  m_2^{\frac{2\alpha}{3}}[H_{\frac{2\alpha}{3}}(m_2)-H_{\frac{2\alpha}{3}}(m_1-1)].
	\end{equation}
	Therefore, for $m_2 \leq M$, $m_2$ can be computed approximately by:
	\begin{equation}
	\label{se2}
	n K'\simeq (m_2-1)^{\frac{2\alpha}{3}} [H_{\frac{2\alpha}{3}}(m_2-1)-H_{\frac{2\alpha}{3}}(m_1-1)].
	\end{equation}
	
	For $\alpha>3/2$:  Using (\ref{se1}), we have
	\begin{equation}
	na(n)K-{(m_1-1)}\simeq  (m_1-1)^{\frac{2\alpha}{3}}\frac{[-(m_1-1)^{1-\frac{2\alpha}{3}}]}{1-\frac{2\alpha}{3}}.
	\end{equation}
	which leads to 
	\begin{equation}
	m_1 \simeq 1+\frac{2\alpha-3}{2\alpha}na(n)K.
	\end{equation}
	Now if $m_2\leq M$, following (\ref{m1/m2}) we have
	\begin{equation}
	m_2\simeq m_1 (a(n))^{-\frac{3}{2\alpha}}\simeq \frac{2\alpha-3}{2\alpha}nK (a(n))^{1-\frac{3}{2\alpha}}.
	\end{equation}
	
	For $\alpha=3/2$: Assuming $m_2\leq M$, and by using (\ref{se2}) and (\ref{m1/m2}), we have
	\begin{equation}
	m_2-1\simeq \frac{nK-{(m_1-1)}a^{-1}(n)}{\log m_2}.
	\end{equation}
	It follows that
	\begin{equation}
	m_2-1\simeq  \frac{nK}{\log m_2}.
	\end{equation}
	This contradicts $m_2=O(n^{\beta})$, where $\beta<1$. Hence $m_2=M+1$. Assuming $m_1>1$, and using (\ref{se1}), we have
	\begin{equation}
	m_1-1\simeq \frac{nKa(n)-{(m_1-1)}}{\log m_2} .
	\end{equation}
	resulting in $m_1=\Theta(\frac{na(n)}{\log n})$.
	
	For $\alpha<3/2$: Assuming $m_2\leq M$, and by using (\ref{se2}), it follows that
	\begin{equation}
	\label{m2alpha}
	\frac{m_2-1}{1-2\alpha/3}\simeq {nK'}.
	\end{equation}
	Clearly, this contradicts the $m_2\leq M$ assumption. Therefore, $m_2=M+1$. Now using (\ref{se1}) we have
	\begin{IEEEeqnarray}{+rCl+x*}
		(m_1-1)^{\frac{2\alpha}{3}}&\simeq & \frac{nKa(n)- (m_1 - 1) }{[H_{\frac{2\alpha}{3}}(m_2-1)-H_{\frac{2\alpha}{3}}(m_1-1)]}\nonumber\\
		&\simeq & \frac{nKa(n)}{M^{1-2\alpha/3}}.
	\end{IEEEeqnarray}
	leading to $m_1= \Theta \left(\frac{(na(n))^{\frac{3}{2\alpha}}}{M^{\frac{3}{2\alpha} -1}}\right)$.

\subsection{Proof of Lemma \ref{m1m2h}}
\label{app_m1m2h}
Since $\mu<1$, $K'\rightarrow K-\frac{(m_1-1)}{2\log n}$ as $n\rightarrow \infty$. By definition, $m_1$ is the smallest index for which the number of holders is less than $a^{-1}(n)-f(n)$. Using (\ref{Xmh}), it follows that
\begin{equation}
\label{f1h}
2K'\log n<m_1^{\frac{2\alpha}{3}}[H_{\frac{2\alpha}{3}}(m_2-1)-H_{\frac{2\alpha}{3}}(m_1-1)].
\end{equation}
Now, if $m_1>1$, attempting to decrease the index $m_1$ by one would result in $$\frac{p_{m_1-1}^{2/3}}{\sum_{j=m_1-1}^{m_2-1}p_j^{2/3}} nK'-f(n)\geq a^{-1}(n)-f(n).$$
Hence, we have
\begin{equation}
\label{f2h}
2K'\log n \geq(m_1-1)^{\frac{2\alpha}{3}}[H_{\frac{2\alpha}{3}}(m_2-1)-H_{\frac{2\alpha}{3}}(m_1-2)].
\end{equation}

For $m_1>1$, an approximation of $m_1$ can be obtained from:
\begin{equation}
\label{se1h}
2K'\log n \simeq m_1^{\frac{2\alpha}{3}}[H_{\frac{2\alpha}{3}}(m_2-1)-H_{\frac{2\alpha}{3}}(m_1-1)].
\end{equation} 

Similarly, by the definition of $m_2$, we know $X_{m_2-1}>0$. Using (\ref{Xmh}), it follows that
\begin{equation}
\label{f3h}
\frac{nK'}{f(n)}>(m_2-1)^{\frac{2\alpha}{3}}[H_{\frac{2\alpha}{3}}(m_2-1)-H_{\frac{2\alpha}{3}}(m_1-1)].
\end{equation}
If $m_2 \leq M$, attempting to increase the index $m_2$ by one would lead to  $$\frac{p_{m_2}^{2/3}}{\sum_{j=m_1}^{m_2}p_j^{2/3}} nK' -f(n)\leq 0.$$
It follows that
\begin{equation}
\label{f4h}
\frac{nK'}{f(n)} \leq  m_2^{\frac{2\alpha}{3}}[H_{\frac{2\alpha}{3}}(m_2)-H_{\frac{2\alpha}{3}}(m_1-1)].
\end{equation}
Therefore, for $m_2 \leq M$, $m_2$ can be computed approximately by:  
\begin{equation}
\label{se2h}
\frac{nK'}{f(n)} \simeq (m_2-1)^{\frac{2\alpha}{3}} [H_{\frac{2\alpha}{3}}(m_2-1)-H_{\frac{2\alpha}{3}}(m_1-1)].
\end{equation}

For $\alpha>3/2$: By using (\ref{se1h}), we have
\begin{equation}
2\log n(K-\frac{(m_1-1)}{2\log n})\simeq  (m_1-1)^{\frac{2\alpha}{3}}\frac{[-(m_1-1)^{1-\frac{2\alpha}{3}}]}{1-\frac{2\alpha}{3}}.
\end{equation}
which leads to 
\begin{equation}
m_1 \simeq 1+\frac{2\alpha-3}{\alpha}K\log n.
\end{equation}
Now if $m_2\leq M$, following (\ref{m1/m2h}) we have
\begin{equation}
m_2= \Theta\left(\left(\frac{n}{f(n)}\right)^{\frac{3}{2\alpha}}(\log n)^{1-\frac{3}{2\alpha}}\right).
\end{equation}  
For $\alpha=3/2$: using (\ref{se1h}), we have
\begin{equation}
m_1-1\simeq 2\log n\left(K-\frac{(m_1-1)}{2\log n}\right).
\end{equation}
leading to \smash{$m_1-1\simeq\frac{K\log n}{\log m_2}$}. Now, if $m_2=M+1$, then $m_1=\Theta(1)$. Otherwise, if $m_2\leq M$, combining this result with (\ref{m1/m2h}), we have $$m_1=\Theta(1).$$
\begin{equation*}
m_2=\Theta\left(\frac{n}{f(n)\log n}\right).
\end{equation*}

For $\alpha<3/2$: using (\ref{f2h}) we have
\begin{equation}
(m_1-1)^{\frac{2\alpha}{3}}\leq \frac{2\log n(K-\frac{(m_1-1)}{2\log n}) }{m_2^{1-2\alpha/3}}.
\end{equation}
Using straightforward calculations, it follows that
\begin{equation}
\label{m1alpha}
(m_1-1)^{\frac{2\alpha}{3}}\leq \frac{2K\log n}{m_2^{1-2\alpha/3}}.
\end{equation}

If $m_2=M+1$ then  clearly, the RHS converges to zero. Therefore, $m_1 \rightarrow 1$ as $n$ grows. Otherwise, if $m_2\leq M$, by using (\ref{se2h}) we have
\begin{equation}
\frac{m_2-1}{1-2\alpha/3}\simeq \frac{n(K-\frac{(m_1-1)}{2\log n})}{f(n)}=\Theta(\frac{n}{f(n)}).
\end{equation}
By plugging in this result in (\ref{m1alpha}), the RHS converges to zero, as previously. Thus, $m_1 \rightarrow 1$ as $n$ grows.

\bibliographystyle{IEEEtran}
\bibliography{ref}
\end{document}